\documentclass[]{article}
\usepackage{graphicx} 
\usepackage{caption}
\usepackage[labelformat=simple]{subcaption}

\usepackage{amsmath,amsfonts,amssymb,amsthm}
\usepackage{bm}
\usepackage{bbm}
\usepackage{url}

\newcommand{\Id}{\operatorname{Id}}

\usepackage{physics}
\usepackage{tikz}
\usetikzlibrary{decorations.pathmorphing}
\usetikzlibrary{decorations.markings,arrows}
\usetikzlibrary{math}
\usetikzlibrary{hobby}
\usepackage[maxcitenames=1,backend=biber,style=chem-angew,sorting=nyt,articletitle=true,doi=false,isbn=false,url=false,eprint=true]{biblatex}
\usepackage[colorinlistoftodos,textsize=tiny]{todonotes}
\usepackage{lineno}
\usepackage{indentfirst}

\usepackage[colorlinks=true,linkcolor=red,citecolor=green]{hyperref}

\newtheorem{theorem}{Theorem}
\newtheorem{proposition}{Proposition}

\newtheorem{lemma}{Lemma}
\theoremstyle{definition}
\newtheorem{remark}{Remark}
\theoremstyle{definition}
\newtheorem{definition}{Definition}

\usepackage{enumitem}

\title{
Diffusion Limit of the Low-Density Magnetic Lorentz Gas 

}

\author{Alessia Nota\footnote{Gran Sasso Science Institute.\\ E-mail: alessia.nota@gssi.it}, Dominik Nowak\footnote{University of Basel, Department of Mathematics and Computer Science.\\ E-mail: dominik.nowak@unibas.ch}, 
Chiara Saffirio\footnote{University of Basel, Department of Mathematics and Computer Science.\\ E-mail: chiara.saffirio@unibas.ch} \footnote{University of British Columbia, Department of Mathematics.\\ E-mail: saffirio@math.ubc.ch}}

\date{\today}

\begin{document}

\maketitle

\begin{abstract}
We consider the magnetic Lorentz gas proposed by Bobylev et al.~\cite{bobylev_two-dimensional_1995}, which describes a point particle moving in a random distribution of hard-disk obstacles in $\mathbb{R}^2$ under the influence of a constant magnetic field perpendicular to the plane. We show that, in the coupled low-density and diffusion limit, when the intensity of the magnetic field is smaller than $\frac{8\pi}{3}$, the non-Markovian effects induced by the magnetic field become sufficiently weak. Consequently, the particle's probability distribution converges to the solution of the heat equation with a diffusion coefficient dependent on the magnetic field and given by the Green-Kubo formula. This formula is derived from the generator of the generalized Boltzmann process associated with the generalized Boltzmann equation, as predicted in~\cite{bobylev_two-dimensional_1995}.
\end{abstract}

\section{Introduction and Main Result}

{\it The magnetic Lorentz gas.} We consider a classical test particle of mass $m=1$ and charge $q=-1$ moving in two dimensions within a bounded, measurable region $\Lambda \subset \mathbb{R}^2$, and subject to a homogeneous, constant magnetic field oriented along the positive $z$-axis, perpendicularly to the plane. 
The charged particle then moves under the action of the Lorentz force $F(v) = - (v \times B)=-v^\perp \abs{B}$ and undergoes elastic collisions with $N$ stationary disks of radius $\varepsilon > 0$, whose centers 
\begin{equation}
    \bm{c}_N = (c_1,...,c_N) \subset \Lambda^N
\end{equation}
are Poisson distributed with intensity $\mu>0$, with $c_i$ being the center of the $i$-th obstacle. The probability of finding $N$ obstacles in the region $\Lambda$ is given by
\begin{equation}\label{eq:Poisson}
    \mathbb{P}(\dd \bm{c}_N) = e^{-\mu \abs{\Lambda}} \frac{\mu^N}{N!} \dd \bm{c}_N,
\end{equation}
where $\dd \bm{c}_N$ is shorthand for $\dd c_1 ... \dd c_N$ and $\abs{\Lambda}$ denotes the Lebesgue measure of $\Lambda$.
 Between collisions, the point particle moves counterclockwise on circular orbits with Larmor radius $R = \abs{v}/\abs{B}$ and angular velocity with magnitude $\widetilde{\Omega} = \abs{B}$. Setting $v \in S_1$ without loss of generality, the cyclotron period is $\widetilde{T}= 2 \pi/ \widetilde{\Omega} = 2\pi/\abs{B}$. In the following, the aforementioned model will be referred to as {\it the magnetic Lorentz gas}.

{\it Scaling regimes.} Let $\varepsilon>0$ be the radius of each disk. 
We rescale the intensity $\mu$ of the obstacles as 
\begin{equation}\label{eq:scaling}
\mu_{\varepsilon}=\varepsilon^{-(d-1)}\eta_\varepsilon\mu,
\end{equation}
which for $d=2$ reads $\mu_\varepsilon=\varepsilon^{-1}\eta_{\varepsilon}\mu$, where, from now on, $\mu>0$ is fixed and $\eta_{\varepsilon}$ is slowly diverging as $\varepsilon\to 0$. More precisely we assume that $\eta_{\varepsilon}$ is such that $\mu_\varepsilon \varepsilon^2\to 0$ and $\mu_\varepsilon \varepsilon\to \infty$ (e.g., cf.~\cite{basile_derivation_2015}, see also~\cite{bodineau_brownian_2016,erdos_quantum_2008} for a similar scaling). We will then denote by $\mathbb{P}^{\varepsilon}$ the rescaled probability density~\eqref{eq:Poisson}, where $\mu$ is replaced by $\mu_\varepsilon$, and by $\mathbb{E}^{\varepsilon}$ the expectation with respect to the measure $\mathbb{P}^{\varepsilon}$.

\noindent We recall that in the standard Boltzmann--Grad regime the intensity $\mu$ of the Poisson distribution is rescaled as $\mu_\varepsilon=\varepsilon^{-1}\mu$ and $\mu_\varepsilon\varepsilon=\ell^{-1}=O(1)$ where $\ell$ is the mean free path that is of order $1$, i.e.~a particle has on average one collision per unit time. Notice that this scaling where $t$, $x$, $v$ are kept fixed, but the radius of any disk and the density of scatterers is rescaled, is equivalent to the scaling in which the radius of any disk is not scaled, but one rescales hyperbolically the space variable $x$ and the time variable $t$ making the density simultaneously vanishing suitably (cf.~\cite{spohn_kinetic_1980}).

\noindent In the scaling limit~\eqref{eq:scaling} we consider in this paper,
the gas is slightly more dense than in the standard Boltzmann--Grad regime discussed above. 
More precisely, the gas is still dilute ($\mu_\varepsilon \varepsilon^2\to 0$ as $\varepsilon\to 0$), but now the mean free path of the tagged particle diverges ($\mu_\varepsilon \varepsilon\to \infty$ as $\varepsilon\to 0$) slow enough to guarantee a dilute configuration of scatterers.
Notice that, according to~\eqref{eq:scaling}, we  implicitly  further rescale the time variable, and consequently also the Larmor time, by a factor $\eta_\varepsilon$. Namely,  $\widetilde{T}\to \eta^{-1}_\varepsilon \widetilde{T}$ and we set ${T}$ to be the rescaled Larmor time \begin{equation}\label{eq:scaled-Larmor}{T}=\eta^{-1}_\varepsilon \widetilde{T}.
\end{equation} 
Alternatively, from the scaling~\eqref{eq:scaling} one sees that the angular velocity $\widetilde{\Omega}$ of the test particle undergoes the transformation $\widetilde{\Omega}\to{\eta_\varepsilon}^{-1}\widetilde{\Omega}$, i.e.~the macroscopic angular velocity $\Omega:=\eta_{\varepsilon}^{-1}\widetilde{\Omega}$ and hence the Larmor time becomes $T=\frac{2\pi}{\eta_\varepsilon\Omega}$ in macroscopic variables.

\noindent Due to the fact that the gas is dilute in the regime~\eqref{eq:scaling}, at mesoscopic scale, the dynamics of the magnetic Lorentz gas will be approximated by the generalized Boltzmann equation~\eqref{eq:LinMagBoltzmann}
with a diverging factor $\eta_\varepsilon$ in front of the collision operator, which corresponds to a regime in which the tagged particle undergoes a diverging number of collisions
per unit of time. This suggests that we can then look at a longer time scale $\eta_\varepsilon t$, i.e. 
\begin{equation}\label{eq:hydro-timescale}
t\to \eta_{\varepsilon}t,
\end{equation}
where, with a slight abuse of notation, we are  now denoting by $t$ the macroscopic time scale.  
On this longer time scale, we capture the diffusive behavior of the system, namely a heat equation 
arises, which provides the hydrodynamics of the magnetic Lorentz gas.    
We emphasize that we can implement this program thanks to the fact that we can provide explicit estimates of the set of pathological  configurations in the kinetic approximation (cf.~\eqref{eq:errorest}),  which allow stronger divergence in time and density, and then to obtain the diffusive behavior of the system. More precisely, the explicit control of the error in the kinetic limit 
suggests the precise scale of times for which the diffusive limit can be achieved. For this reason we will assume that as $\varepsilon \to 0$, $\eta_\varepsilon$ diverges in such a way that 
\begin{equation}\label{eq:cdteta}
\varepsilon^{\frac{1}{2}}\eta_\varepsilon^5\to 0.
\end{equation}
This is indeed dictated by the error estimate~\eqref{eq:errorest} below: imposing~\eqref{eq:cdteta}, we guarantee that the error in~\eqref{eq:errorest} is still vanishing on the long time scale $\eta_\varepsilon t$ as $\varepsilon\to 0$.

\noindent The aim of this paper is to study the long-time behavior of the magnetic Lorentz gas in the low-density limit~\eqref{eq:scaling} by giving a rigorous proof of the behavior conjectured in~\cite{bobylev_two-dimensional_1995}. 

{\it Kinetic description.} In absence of the magnetic field, Gallavotti~\cite{gallavotti_divergences_1969} proved that, in the Boltzmann--Grad limit, i.e.~$\mu_{\varepsilon}\varepsilon=O(1)$, the test particle's dynamics can be approximated by the Lorentz kinetic equation, also referred to as linear Boltzmann equation. However, as first observed in~\cite{bobylev_two-dimensional_1995}, the presence of the external magnetic field introduces new interesting memory effects, which are illustrated by the following simple computation. Let $\mathcal{C}$ be a configuration such that no obstacles are located in an annulus $\mathcal{A}^\varepsilon(R)$ defined by the radii $R-\varepsilon$ and $R +\varepsilon$. The probability that the test particle performs a complete cyclotron orbit without hitting an obstacle is then given by
\begin{equation} \label{eq:Prob Circling}
    \mathbb{P}^\varepsilon (\mathcal{C}) = e^{- \mu_\varepsilon \abs{\mathcal{A}^\varepsilon(R)}} = e^{-4 \pi R\mu\eta_\varepsilon} = e^{-2\mu\eta_\varepsilon\widetilde{T}}=e^{-2\mu T},
\end{equation}
where in the last identity we used~\eqref{eq:scaled-Larmor}.
  As $\varepsilon \to 0$, the probability that the test particle is trapped in a cyclotron orbit is non-vanishing. Hence, the standard linear Boltzmann equation cannot describe the model's kinetic behavior. The simple reasoning in~\eqref{eq:Prob Circling} suggests also a non-vanishing probability that the test particle encounters multiple subsequent collisions with the same obstacle, referred to as self-recollisions. This non-Markovian behavior has been studied in~\cite{bobylev_two-dimensional_1995,bobylev_there_1997} (see also~\cite{kuzmany_magnetotransport_1998}) and rigorously derived in the Boltzmann--Grad limit (i.e.~$\mu_\varepsilon\varepsilon=O(1)$) in~\cite{nota_two-dimensional_2022}, where it was shown that the one-particle correlation function $f^\varepsilon$ in~\eqref{eq:fep} converges in $L^1$ to a solution of the \textit{generalized Boltzmann equation}
\begin{equation} \label{eq:LinMagBoltzmann}
    \begin{aligned}
	D_t f^{G}(t,x,v) &= \mu \sum_{k=0}^{[t/T]} e^{-2 k T}\int_{S_{1}}  (v \cdot n)_{+} \left \{\sigma_{n}-1\right\} f^{G}(t-k T, S^{(k)}_{n}(x,v) )\,\dd n\\
    &=: \mathcal{L}^G f^{G}(t,x,v),
    \end{aligned}
\end{equation}
where $D_t:= (\partial_t+  v \cdot \nabla_x -(v \times B) \cdot \nabla_v)$ denotes the material derivative, $n \in S_{1}$ is the scattering vector and $[w]$ takes the integer part of $w$. 
The operator $\sigma_{n}$ acts as
\begin{equation}
	\sigma_{n} F(x,v) =F(S_n^{(1)}(x,v))= F(x,v')= F(x,v-2(v\cdot n)n)
\end{equation}
for any function $F$, and $v'$ denotes the post-collisional velocity.
Furthermore, for $k>1$, $S^{(k)}_{n}(x,v) = (x,R_{k\vartheta}(v))$, where
\begin{equation}
	R_{\vartheta} = 
	\begin{pmatrix}
	\cos(\vartheta) & - \sin(\vartheta)\\
	\sin(\vartheta) & \cos(\vartheta)
	\end{pmatrix}
	.
\end{equation}
The density $f^{G}$ takes into account the probability that an electron is \textit{circling} (see Fig.~\ref{fig:circling}) or \textit{wandering} (see  Fig.~\ref{fig:wandering}).
\begin{figure}
     \centering
     \begin{subfigure}[t]{0.45\textwidth}
         \centering
         \includegraphics[width=\textwidth]{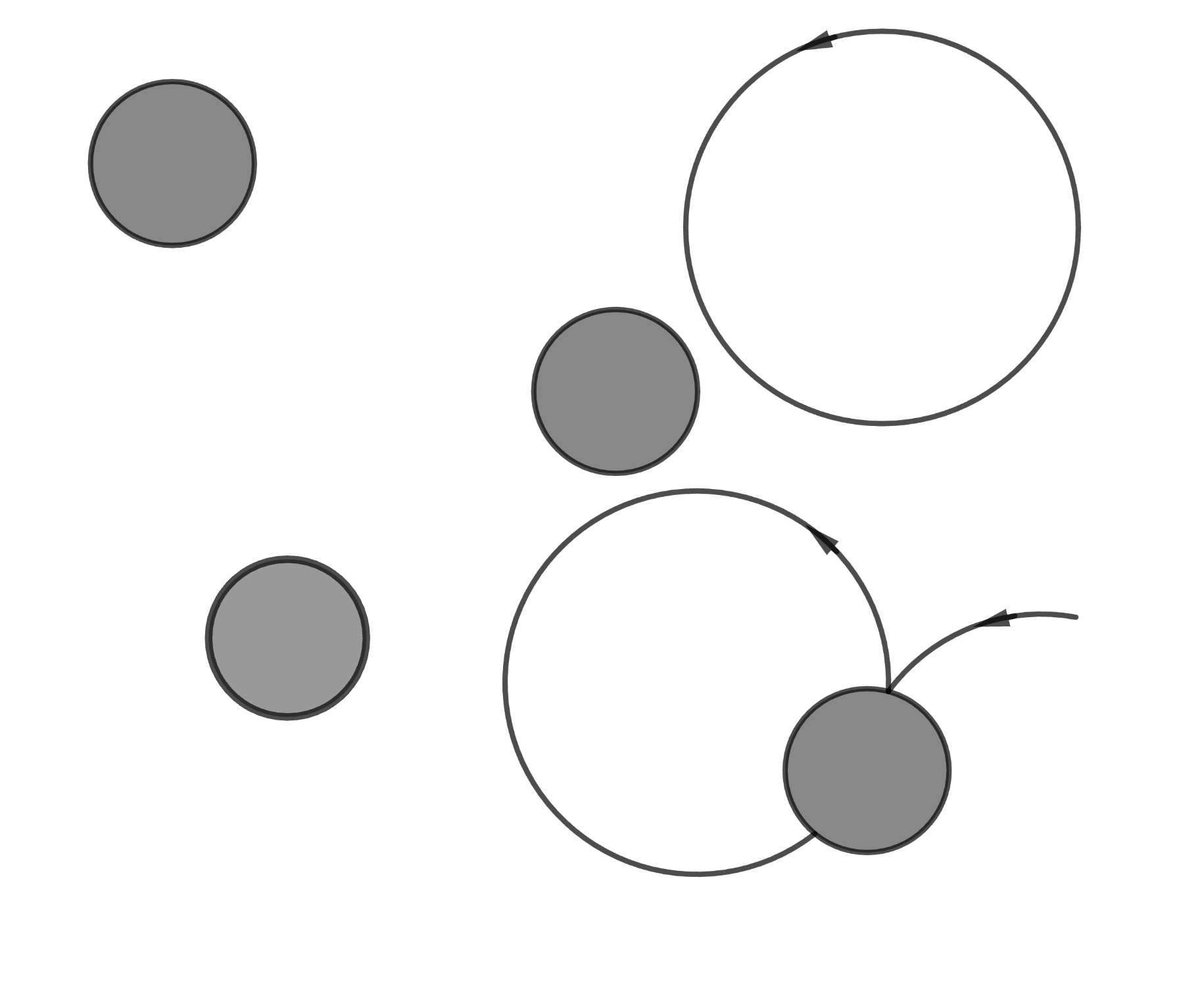}
         \caption{Circling light particle}
         \label{fig:circling}
     \end{subfigure}
     \hfill
     \begin{subfigure}[t]{0.45\textwidth}
         \centering
         \includegraphics[width=\textwidth]{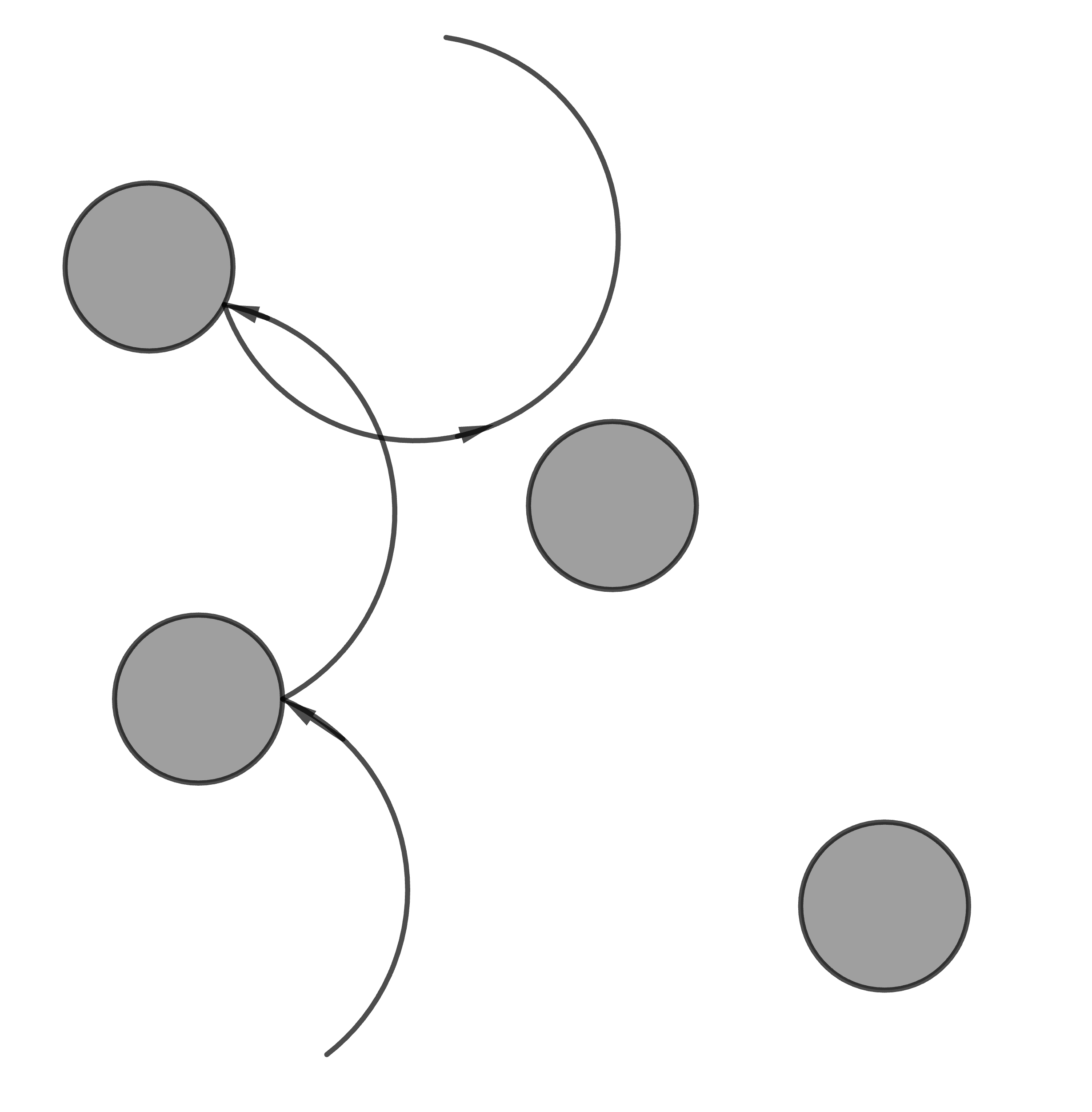}
         \caption{Wandering light particle}
         \label{fig:wandering}
     \end{subfigure}
        \caption{Realisations of possible trajectories}
        \label{fig:PossibleTraj}
\end{figure}
Clearly, 
\begin{equation}
f^{G}(t,x,v) =
	\begin{cases}
	f^{C}(t,x,v)+f^{W}(t,x,v), & 0<t<T\\
	f^{W}(t,x,v), &t>T ,
	\end{cases}
\end{equation}
where the superscript $C$ and $W$ denotes the fraction of circling and wandering electrons, respectively. 
Since $\int_{\mathbb{R}^2 \times S_{1}} f^{G}(t,x,v) \dd x \dd v = 1$ for any $t$, we have
\begin{equation}
\int_{\mathbb{R}^2\times S_{1}} f^{G}(t,x,v) \dd x \dd v =
\begin{cases}
1,& 0<t<T\\
1-e^{-2T},& t>T
\end{cases}
\end{equation}
when taking into account the fraction of circling electrons for $t>T$.\\
Equation~\eqref{eq:LinMagBoltzmann} suggests a splitting of the operator into a Markovian, as well as a non-Markovian part. The non-Markovianity appears since self-recollisions are admissible and are taken into account by the sum over $k$ weighted by the factor $e^{-2kT}$. We find that the right hand side of~\eqref{eq:LinMagBoltzmann} can be rewritten by grouping the sum as follows
\begin{equation}\label{eq:LinMagBoltzmann2}
    \begin{aligned}
        \mathcal{L}^G f^{G} = &\mu \sum_{k=0}^{[t/T]} e^{-2 k T}\int_{S_{1}}  (v \cdot n)_{+} \left \{\sigma_{n}-1\right\} f^{G}(t-k T, S^{(k)}_{n}(x,v) )\,\dd n\\
        = &\mu \int_{S_1} (v\cdot n)_+ [f^G(v') - f^G(v)]\,\dd n\\
        &+ \mu \sum_{k=1}^{[t/T]} e^{-2 k T}\int_{S_{1}}  (v \cdot n)_{+} \left \{\sigma_{n}-1\right\} f^{G}(t-k T, S^{(k)}_{n}(x,v) )\,\dd n\\
        =:& \mathcal{L} f^G + \mathcal{M}f^G.
    \end{aligned}
\end{equation}
Indeed for $k=0$, we recover the linear Boltzmann operator in absence of $B$
\begin{equation}\label{eq:L-operator}
    \mathcal{L} f^G = \mu \int_{S_1} (v\cdot n)_+ [f^G(v') - f^G(v)]\,\dd n,
\end{equation} 
while $\mathcal{M}$ is the contribution due to the presence of the external magnetic field 
\begin{equation}\label{eq:M-operator}
    \mathcal{M} f^{G}(t,x,v) := \mu \sum_{k=1}^{[t/T]} e^{-2kT} \int_{S_1}  (v\cdot n )_+ \left \{\sigma_n - 1\right\} f^G(t-kT,S_n^{(k)}(x,v))\,\dd n,
\end{equation}
which contains the memory terms.\\
Notice that the rescaling of the Larmor time does not affect the generalized Boltzmann collision operator $\mathcal{L}^G$, which is homogeneous in $t$ and $T$.

{\it Main result.} Our main result builds on this splitting into Markovian and non-Markovian operators. On the hydrodynamic time scale described in~\eqref{eq:scaling} and~\eqref{eq:hydro-timescale}, we prove that the dynamics of the moving particle is given by the heat equation. The main novelty consists in the form of the diffusion coefficient (see~\eqref{eq:GreenKubo}-\eqref{eq:D-Split} below), given by the Green-Kubo formula associated to the operator $\mathcal{L}^G$ in~\eqref{eq:LinMagBoltzmann}, which keeps track of the non-Markovian terms induced by the external magnetic field. More precisely, for a given configuration of obstacles $\bm{c}_N$, we denote by $\gamma^{t}_{\bm{c}_{N},\varepsilon}(x,v)$, $t \in \mathbb{R}$ the (backward) Hamiltonian flow, solution to the Newton equations for a particle with initial state $(x,v)$, in a given sample $\bm{c}_N$ of obstacles of radius $\varepsilon$. For a given initial datum $f_0=f_0(x,v)$, the particle distribution at time $t > 0$ is
\begin{equation}\label{eq:fep}
f^{\varepsilon}(t,x,v):=\mathbb{E}^{\varepsilon}[f_0(\gamma^{-t}_{\bm{c}_{N},\varepsilon}(x,v))\, \mathbbm{1}_{\{\min_i |x-c_i| >\varepsilon\}}]\;.
\end{equation}
Then our main result reads

\begin{theorem} \label{thm:MainTheorem}
    Let $f_0 \in C_c(\mathbb{R}^2 \times S_1)$ be a compactly supported continuous probability density and assume that all first and second derivatives in $x$ and $v$ of $f_0$ are bounded.  
    Let $f^\varepsilon$ be defined as in~\eqref{eq:fep}.  
    Fix a positive number $\bar{t}\in\mathbb{R}$ and let {$\mu_\varepsilon = \varepsilon^{-1}\eta_\varepsilon \mu$} and 
    ${h}^\varepsilon(t,x,v):= f^\varepsilon({t \eta_\varepsilon}, x,v)$. Then, for $\abs{B} \in \left[0,\frac{8 \pi}{3}\right)$ and all $t \in [0,\bar{t}\,)$, 
    \begin{equation}
        \lim_{\varepsilon \to 0} \norm{{h}^\varepsilon(t,\cdot,\cdot) - \rho(t,\cdot)}_{L^2(\mathbb{R}^2 \times S_1)} = 0,
    \end{equation}
    where $\rho = \rho(t,x)$ solves the following heat equation 
    \begin{equation} \label{eq:HeatEquation}
        \begin{cases}
            \partial_t \rho = D_B \Delta_x \rho\\
            \rho(0,x) = \langle f_0 \rangle,
        \end{cases}
    \end{equation}
    with $\langle f_0 \rangle =: \frac{1}{2 \pi} \int_{S_1} f_0 (x,v) \dd v$ and the diffusion coefficient $D_B$ given by the Green-Kubo formula
    \begin{equation} \label{eq:GreenKubo}
        D_B =  \frac{1}{2 \pi} \int_{S_1} v \cdot \left(-\mathcal{L}^G \right)^{-1} v \dd v = \int_0^\infty \mathbb{E} \left[ v \cdot V(t,v)\right] \dd t,
    \end{equation}
    where $V(t,v)$ is the stochastic process generated by $\mathcal{L}^G$ starting from $v$ and $\mathbb{E}[\cdot]$ denotes the expectation with respect to the invariant measure on $S_1$.

\noindent    Furthermore, 
    \begin{equation} \label{eq:D-Split}
        D_B =  \frac{1}{2 \pi} \int_{S_1} v \cdot (-\mathcal{L})^{-1} v \dd v +  \frac{1}{2 \pi}\sum_{k=1}^\infty \int_{S_1} v\cdot (-\mathcal{L})^{-1} \left[\mathcal{M}(-\mathcal{L})^{-1}\right]^k v\,\dd v,
    \end{equation}
where $\mathcal{L}$ and $\mathcal{M}$ are defined in~\eqref{eq:L-operator} and~\eqref{eq:M-operator}, respectively.
\end{theorem}

\noindent Some remarks are in order. 

\begin{remark}
 We stress that our result provides a direct limit from the microscopic dynamics to the hydrodynamic scale, using the kinetic scale as a bridge between the two. 
 \end{remark}

\begin{remark} As pointed out in~\cite{bobylev_two-dimensional_1995, kuzmany_magnetotransport_1998}, a large magnetic field would prevent the diffusive behavior to take place. For $\abs{B}=\infty$, this follows easily from equation~\eqref{eq:Prob Circling}, recalling that $T=\frac{2\pi}{\abs{B}}$. Indeed, by formally taking $\abs{B}=\infty$ in~\eqref{eq:Prob Circling}, the particle is trapped in a periodic orbit with probability one, hence the generalized Boltzmann equation does not give a good approximation of such situation and therefore no diffusion takes place at the hydrodynamic scale (see also Remark~\ref{rem:B-infty}). As conjectured in~\cite{bobylev_two-dimensional_1995}, there should exist a critical value $B_c$ such that for $\abs{B} >B_c$ there is no diffusive behavior, whereas for $\abs{B} <B_c$ the hydrodynamics is given by equation~\eqref{eq:HeatEquation}. With our techniques, we are at the moment able to prove the diffusive behavior only for $\abs{B}<\frac{8\pi}{3}$. From our approach, it is not clear whether $B_c=\frac{8\pi}{3}$, as the constraint $\abs{B}<\frac{8\pi}{3}$ comes as a technical assumption to ensure that the operator $\mathcal{L}^G$ is invertible (see Lemma~\ref{lem: Inversion LG} below).  
\end{remark}

\begin{remark}
 Although our result considers densities that are slightly more dense than the usual Boltzmann--Grad regime, our scaling is still compatible with a rarefied situation. Higher densities have been the object of numerical studies in~\cite{bobylev_two-dimensional_1995, kuzmany_magnetotransport_1998}, where the authors find critical values of the density for which the moving particle stays trapped in a cage formed by overlapping obstacles and by the magnetic field. However, an analytical proof is still missing and will be the object of future study.\end{remark}
\begin{remark} The form of the diffusion coefficient in~\eqref{eq:D-Split} reveals the nature of the underlying process: the first term accounts for the Markovian part, while the second introduces memory effects of unbounded order, depending on the strength of the magnetic field, thereby supporting the conjecture proposed in~\cite{bobylev_two-dimensional_1995}. Moreover, 
our findings are in agreement with the fact that, for $B=0$, the diffusion coefficient should coincide with the formula obtained in the standard hydrodynamic limit of the Lorentz gas (see~\cite{esposito_chapter_2005}). This can be easily seen by taking the limit $B\to 0$ in~\eqref{eq:D-Split}. We also mention that the diffusion coefficient appears in front of the Laplacian in~\eqref{eq:HeatEquation}, modulating the diffusion process depending on the intensity of the magnetic field. This is most visible from the explicit expression of $D_B$ given in~\eqref{eq:D-Split}.
\end{remark}
\begin{remark}\label{remark:lutsko-toth} Recently, Lutsko and T\'oth~\cite{lutsko_diffusion_2024} studied the magnetic Lorentz gas and proved that, by first applying the Boltzmann--Grad limit and then taking the limit as time goes to infinity, the rescaled trajectory converges to a Brownian motion modulated by a $B$-dependent random variable $\alpha$. When $\abs{B}=\infty$, $\alpha=0$ with probability one, and the rescaled trajectory converges to zero. This follows by a highly non-trivial adaptation of the coupling method developed by the same authors in~\cite{lutsko_invariance_2020}. Their result captures the absence of diffusion for $\abs{B}=\infty$, as well as the qualitative deterioration of the diffusion process as $B$ increases, in the spirit of~\cite{bobylev_two-dimensional_1995}. Our result differs from the one in~\cite{lutsko_diffusion_2024} in the following aspects: while it does not capture the concentration phenomena for large, finite magnetic fields due to the formation of trapping orbits, it gives an explicit expression for the diffusion coefficient $D_B$ through the Green-Kubo formula, highlighting the dependence of $D_B$ on the magnetic field. 
Furthermore, we stress that the techniques used and the scaling regime considered are different and provide complementary information on the hydrodynamics.
\end{remark}

{\it State of the art.} 
The Lorentz gas is a simple but highly non-trivial model which has been proposed by H. A. Lorentz in 1905 to explain the motion of electrons in metals. It represents a rare source of exact results in kinetic theory, providing a concrete example in which microscopic reversibility can be reconciled with macroscopic irreversibility.
Indeed, from this model, one can prove, under suitable scaling limits, a validation of linear kinetic equations, and, from these, of diffusion equations. 
The first rigorous mathematical result in this direction was proved in~\cite{gallavotti_divergences_1969}, where the linear Boltzmann equation has been derived in the Boltzmann--Grad regime.
We refer to ~\cite{spohn_lorentz_1978, lebowitz_transport_1978,boldrighini_boltzmann_1983,desvillettes_linear_1999,basile_derivation_2015,nota_diffusive_2015,lutsko_invariance_2020, nota_theory_2018} for related results and later developments.

The presence of a given external field strongly affects the derivation of the linear Boltzmann equation in the low-density limit. In~\cite{ 
bobylev_two-dimensional_1995} and later in~\cite{bobylev_there_1997,bobylev_liouville_2001}, it has been shown that the motion of a test particle in a plane with a Poisson distribution of hard disks and a uniform and constant magnetic field perpendicular to the plane heuristically leads to the generalized Boltzmann equation with memory terms. We also mention that this model has been studied numerically in~\cite{kuzmany_magnetotransport_1998}, where several regimes have been identified. 
A rigorous derivation of this non-Markovian Boltzmann equation has been recently obtained in~\cite{nota_two-dimensional_2022}.
We also refer to~\cite{marcozzi_derivation_2016} where linear kinetic equations with a magnetic transport term have been derived, but the non-Markovian behavior of the 
limit process disappears considering a slightly different magnetic Lorentz gas. 

The rigorous derivation of hydrodynamic equations, specifically the heat equation, from the mechanical system given by the random Lorentz gas is  a difficult and still unsolved problem. 
It is worth mentioning that the rigorous validity of such diffusive limit has been obtained in Bunimovich and Sinai (cf.~\cite{bunimovich_statistical_1981}) when the scatterers are periodically distributed. 

Nonetheless, in the random setting, one can handle this problem by deriving the diffusion equation relying on the kinetic
approximation of the microscopic dynamics. More precisely, the heat equation can be obtained from the random Lorentz gas using as a bridge the kinetic equation, which arises in a suitable kinetic limit.  As a consequence, the diffusion coefficient $D$ is given by the Green–Kubo formula associated to the kinetic equation. 
This strategy works once we have an explicit control of the error in the kinetic limit, which
suggests the scale of times for which the diffusive limit can be achieved. This idea has been used to obtain the heat equation in different contexts, see
\cite{bodineau_brownian_2016,basile_derivation_2015,basile_diffusion_2014}, and also~\cite{erdos_quantum_2008}. We also refer to~\cite{lutsko_invariance_2020, lutsko_diffusion_2024} for a different approach in this direction. 

In this paper, we consider the magnetic Lorentz gas and prove that, in the joint low-density and diffusive limit, the heat equation naturally appears, in the same spirit of~\cite{bodineau_brownian_2016,basile_derivation_2015,basile_diffusion_2014}. Here we rely on the generalized Boltzmann equation as a kinetic
approximation of the microscopic dynamics. The main novelty consists in the form of the diffusion coefficient (see~\eqref{eq:GreenKubo}-\eqref{eq:D-Split}), given by the Green-Kubo formula associated to the operator $\mathcal{L}^G$ in~\eqref{eq:LinMagBoltzmann}, which keeps track of the non-Markovian terms induced by the external magnetic field. Recently, also Lutsko and T\'oth~\cite{lutsko_diffusion_2024} considered this problem in a slightly different regime. For a comparison between~\cite{lutsko_diffusion_2024} and our main result, see Remark~\ref{remark:lutsko-toth} above.

\section{Strategy}\label{sec:strategy}

For a given configuration of obstacles $\bm{c}_N$, we recall that $\gamma^{t}_{\bm{c}_{N},\varepsilon}(x,v)$, $t \in \mathbb{R}$ is the (backward) Hamiltonian flow, solution to the Newton equations for a particle with initial state $(x,v)$, in a given sample $\bm{c}_N$ of obstacles of radius $\varepsilon$. We also recall that, for a given initial datum $f_0=f_0(x,v)$, the particle distribution at time $t > 0$ is
\begin{equation}
f^{\varepsilon}(t,x,v):=\mathbb{E}^{\varepsilon}[f_0(\gamma^{-t}_{\bm{c}_{N},\varepsilon}(x,v))\, \mathbbm{1}_{\{\min_i |x-c_i| >\varepsilon\}}]\;.
\end{equation}

\noindent Let  $\bar{f}^{G,\varepsilon}$   be the solution of the following problem
 \begin{equation} \label{eq:etaGBE}
        \begin{cases}
             (\partial_t+  v \cdot \nabla_x -(v \times B) \cdot \nabla_v)\bar{f}^{G,\varepsilon}=\eta_\varepsilon \mathcal{L}^G \bar{f}^{G,\varepsilon}(t,x,v)\\
            \bar{f}^{G,\varepsilon}(0,x,v) =   f_0(x,v),
        \end{cases}
    \end{equation} 
    where $f_0\in C_c(\mathbb{R}^2\times S_1)$ is a probability density and $\mathcal{L}^G$ is as in~\eqref{eq:LinMagBoltzmann}-\eqref{eq:LinMagBoltzmann2}. 
    
\noindent It has been proved in~\cite{nota_two-dimensional_2022} that in the classical Boltzmann--Grad limit, namely when $\mu_\varepsilon \varepsilon=\ell^{-1}=O(1)$,  the Lorentz process converges to the generalized Boltzmann process in the annealed setting and, as a corollary, that 
for any probability density $f_0\in C_c(\mathbb{R}^2\times S_1)$, the particle  distribution $f^\varepsilon(t)$ defined in~\eqref{eq:fep} satisfies
\begin{equation} \label{eq:kinlimBG}
\lim_{\varepsilon\to 0}\| f^\varepsilon(t)-f(t)\|_{L^1(\mathbb{R}^2\times S_1)}=0 ,
\end{equation}
where $f$ is the unique mild solution in $L^1(\mathbb{R}^2\times S_1)$ to~\eqref{eq:LinMagBoltzmann}, or, equivalently to~\eqref{eq:etaGBE} with $\eta_\varepsilon=1$.  
  More precisely, the proof of~\eqref{eq:kinlimBG}  is based on an extension of the original argument proposed in~\cite{gallavotti_divergences_1969}.  The goal is to compare the series solution to~\eqref{eq:LinMagBoltzmann} with~\eqref{eq:fep}. This can be achieved by applying a monotonicity argument for which it is enough to bound~\eqref{eq:fep} from below, and exploiting the mass conservation property. The key ingredient in the approximation is  
a suitable change of
variables, enabling the parametrization of the obstacle centers in terms of scattering times and
scattering vectors and neglecting all configurations that do not lead to the limiting solution
of~\eqref{eq:LinMagBoltzmann} as $\varepsilon\to 0$. The proof proceeds by showing that the flow associated with the particle system for $\varepsilon >0$ 
 is stable with respect to perturbations of its argument and converges to the flow described by~\eqref{eq:LinMagBoltzmann} as $\varepsilon \to 0$. We emphasize that, unlike the case without an external magnetic field, the presence of non-vanishing memory terms in the limit requires careful handling. These terms are geometrically characterized by self-recolliding trajectories, and therefore need additional care in the analysis. \\
 In this paper we build on~\cite{nota_two-dimensional_2022} to study the hydrodynamics of the magnetic Lorentz gas. More precisely, exploiting the fact that in the scaling limit~\eqref{eq:scaling} the gas is still dilute and by arguing as in~\cite{nota_two-dimensional_2022}, we can show that the particle density distribution $f^\varepsilon(t)$ defined in~\eqref{eq:fep} approximates the unique mild solution $\bar{f}^{G,\varepsilon}$  to~\eqref{eq:etaGBE}, namely
\begin{equation}\label{eq:kinlim} 
\lim_{\varepsilon\to 0}\| f^\varepsilon(t)-\bar{f}^{G,\varepsilon}(t)\|_{L^1(\mathbb{R}^2\times S_1)}=0. 
\end{equation} 
Furthermore, in this paper we obtain an explicit estimate for the error in the approximation~\eqref{eq:kinlim}, specifically $\abs{E^\varepsilon(t)} \leq C \varepsilon^{\frac{1}{2}} \eta_\varepsilon^3 t^2$. This will be the content of Section~\ref{sec:errors}, see Proposition~\ref{prop:ErrorEstimates}.
It is possible to show that from~\eqref{eq:etaGBE} in the limit $\eta_\varepsilon \to \infty$ one gets a trivial result (cf. Proposition~\ref{prop:ConvergenceToAverage}). Therefore, to obtain a diffusion equation in the limit, one has to look at the solution on a longer time scale, namely $\eta_\varepsilon t$. Denoting by $h^\varepsilon(t,x,v):=\bar{f}^{G,\varepsilon}(\eta_\varepsilon t,x,v)$ where $\bar{f}^{G,\varepsilon}$ solves~\eqref{eq:etaGBE}, it follows that $h^\varepsilon$ solves
\begin{equation} 
\label{eq:TimeRescaledBoltzmann}
        (\partial_t + \eta_\varepsilon v \cdot \nabla_x - \eta_\varepsilon (v \times B) \cdot \nabla_v)h^\varepsilon(t,x,v) = \eta_\varepsilon^2 \mathcal{L}^G h^\varepsilon(t,x,v)\, .
\end{equation}
We further remark that the explicit error estimate obtained in Proposition~\ref{prop:ErrorEstimates}  on the longer time scale $\eta_\varepsilon t$ becomes 
\begin{equation}\label{eq:errorlongtimes}
\abs{E^\varepsilon(\eta_\varepsilon t)} \leq C \varepsilon^{\frac{1}{2}} \eta_\varepsilon^5 t^2,
\end{equation}
which tends to zero in the limit $\varepsilon\to 0$ under assumption~\eqref{eq:cdteta}.  
Then, using a suitable adaptation of the classical Hilbert expansion,  we will prove in Section~\ref{sec:Proof} that the solution $h^\varepsilon$ to~\eqref{eq:TimeRescaledBoltzmann} converges, as $\varepsilon\to 0$, to the solution $\rho(t,\cdot)$ of the heat equation in~\eqref{eq:HeatEquation}.

\section{Preliminary Estimates: The Kinetic Time Scale}\label{sec:errors}
For the reader's convenience, we repeat again the scaling introduced before: The intensity of the Poisson distribution of obstacles will be such that
\begin{equation}
    \mu_\varepsilon = \varepsilon^{-1} \eta_\varepsilon \mu,
\end{equation}
where $\eta_\varepsilon$ diverges such that $\mu_\varepsilon \varepsilon^2 \to 0$ and $\mu_\varepsilon \varepsilon \to \infty$. Moreover, it is worth mentioning that at this point we leave the time variable invariant and therefore all provided estimates are valid for the kinetic time scale (see Proposition~\ref{prop:ErrorEstimates} and compare with Remark~\ref{rem:kinetic-to-longer}).

The convergence of the test particle's dynamics to the generalized Boltzmann equation was proved in~\cite{nota_two-dimensional_2022} without explicit rate of convergence. In this section we collect estimates on the error terms, which make the convergence rate explicit, as shown in the following proposition. From now on, we assume for notational simplicity that the Lebesgue measure of the region $\Lambda$, in which we confine the dynamics, is $|\Lambda|=1$.

\begin{proposition} \label{prop:ErrorEstimates}
    The particle density function $f^\varepsilon$ is close to  $\bar{f}^{G,\varepsilon}$ in $L^1(\mathbb{R}^2 \times S_1)$ for any $t \in [0,\bar{t}]$. To be more precise 
    \begin{equation}
        \norm{f^\varepsilon(t)- \bar{f}^{G,\varepsilon}(t)}_{L^1(\mathbb{R}^2 \times S_1)}= E^\varepsilon(t),
    \end{equation}
    where
    \begin{equation}\label{eq:errorest}
     \abs{E^\varepsilon(t)} \leq C \varepsilon^{\frac{1}{2}} \eta_\varepsilon^3 t^2
    \end{equation}
    for a constant $C>0$ independent of $t$ and $\bar{f}^{G,\varepsilon} = \bar{f}^{G,\varepsilon}(t,x,v)$ is the unique mild solution of the generalized linear Boltzmann equation in~\eqref{eq:etaGBE} with initial datum $f_0$.
\end{proposition}
\begin{remark} \label{rem:kinetic-to-longer}
    Proposition~\ref{prop:ErrorEstimates} holds true for longer time scales. In particular, for $t \to \eta_\varepsilon t$ the error term $E^\varepsilon$ vanishes if $\eta_\varepsilon$ is chosen such that $ \varepsilon^{\frac{1}{2}} \eta_\varepsilon^5 \to 0$ as $\varepsilon \to 0$.
\end{remark}
\begin{remark}
   The property of the Poisson distribution with intensity $\mu_\varepsilon = \varepsilon^{-1} \eta_\varepsilon \mu$ guarantees 
    \begin{equation}
        \mathbb{E}^\varepsilon \left[\mathbbm{1}_{\left\{ \min_i\abs{x-c_i}\leq \varepsilon \right\} }\right] \leq C \varepsilon \eta_\varepsilon,
    \end{equation}
    which implies that configurations in which the tagged particle starts inside (or at the boundary of) an obstacle are negligible in the limit $\varepsilon \to 0$. We omit the proof and refer to~\cite{desvillettes_linear_1999} Lemma~3.1 for the detailed calculation.
\end{remark}
\begin{remark}
    The estimates in order to prove Proposition~\ref{prop:ErrorEstimates} are performed without considering the initial probability density $f_0$. Indeed, including $f_0$ would not change how $E^\varepsilon$ depends on $\varepsilon$, since $f_0$ is bounded and its $L^\infty$-norm would be absorbed into the constant $C>0$.
\end{remark}
\noindent The proof of Proposition~\ref{prop:ErrorEstimates} follows by combining the explicit estimates given in Proposition~\ref{prop:Daisies}, Lemma~\ref{lem:Recollisions} and~\ref{lem:Interferences} as well as the result in~\cite{nota_two-dimensional_2022}.

 However, before starting with the estimates, we want to introduce some terminology and fix the notation. We start with
\begin{definition}[The Magnetic Lorentz/Particle Process]
    For a given $\varepsilon > 0$, initial datum $(x_0,v_0) \in \mathbb{R}^2 \times S_1$ as well as configuration of obstacles $\bm{c}_N$, where $\min_i \abs{c_i - x_0} > 0$, we call the test particle's trajectory a realisation of the magnetic Lorentz (or particle) process starting from $(x_0,v_0)$. We will denote the magnetic Lorentz process by the pair
    \begin{equation}
        \left(\xi^\varepsilon(s),\zeta^\varepsilon(s)\right) := \gamma^{s}_{\bm{c}_{N},\varepsilon}(x_0,v_0), \qquad s \in \mathbb{R},
    \end{equation}
    where the Hamiltonian flow $\gamma^{s}_{\bm{c}_{N},\varepsilon}$ is defined almost surely on the phase space given by the following set
    \begin{equation}
        \mathcal{S}^\varepsilon_{\bm{c}_N} := \left \{ (x,v)\in \mathbb{R}^2 \times S_1 : \min_i \abs{x- c_i} > \varepsilon \right \}.
    \end{equation}
\end{definition}
It is obvious that not every obstacle is hit by the light particle (see Fig.~\ref{fig:PossibleTraj} for an example). This leads us to 
\begin{definition}[Internal and External Obstacles]
    We call an obstacle centered at $c_i$ internal if it is hit by the test particle:
    \begin{equation}
        \inf_{t \geq 0} \abs{\xi^\varepsilon (t) - c_i} = \varepsilon.
    \end{equation}
    The obstacles centered at $c_i$ are external if the trajectory keeps distance larger than $\varepsilon$. More precisely,
    \begin{equation}
        \inf_{t \geq 0} \abs{\xi^\varepsilon(t)- c_i} > \varepsilon.
    \end{equation}
\end{definition}
Given a collection of internal obstacles, it is possible to describe their locations $\{c_i\}$ in terms of the associated scattering data. We introduce 
\begin{definition}[Impact Times and Impact Vectors] \label{def:ScatteringData}
    Given an internal obstacle centered at $c_i$. The impact (or scattering) time $t_i$ is defined as 
    \begin{equation}
        t_i := \sup \left \{s > 0 : \inf_{0 \leq t \leq s} \abs{\xi^\varepsilon(t) - c_i} > \varepsilon \right \}.
    \end{equation}
    In analogy to $t_i$, we call the quantity $ \tau_i$ satisfying
    \begin{equation}
        \tau_i := \inf \left \{s > 0 : \inf_{0 \leq t \leq s} \abs{\xi^\varepsilon(t) - c_i} > \varepsilon \right \}
    \end{equation}
    the exit time. The impact vector $n_i$ of an internal obstacle is given by
    \begin{equation}
        n_i := \frac{\xi^\varepsilon(t_i)- c_i}{\varepsilon} \in S_1.
    \end{equation}
    For a given impact vector $n_i$ and velocity $\zeta^\varepsilon(t_i^+)$, we consider the formed angle $\phi_i$ by $n_i$ and $-\zeta^\varepsilon(t_i^+)$. The so-called impact parameter $b_i$ is defined as
    \begin{equation}
        b_i := \varepsilon \sin \phi_i.
    \end{equation}
    We say that the $n \geq 1$ internal obstacles are ordered if their associated impact times satisfy $t > t_1 > t_2 > ... > t_n > 0$.
\end{definition}

\noindent For the initial datum $(x_0,v_0) \in \mathbb{R}^2 \times S_1$ at $t>0$, we define the generalized Boltzmann process without collisions by 
\begin{equation}
    \left(\xi_0(s),\zeta_0(s)\right) =\gamma^{-t+s}_0(x_0,v_0), \qquad s \in \mathbb{R},
\end{equation}
where $\gamma^{-t+s}_0$ is the Hamiltonian flow solving the following equations of motion
\begin{equation}
    \begin{cases}
        \frac{\dd}{\dd s}\xi_0(s) = \zeta_0(s)\\
        \frac{\dd}{\dd s}\zeta_0(s) = - \zeta_0^\perp(s) \abs{B}\\
        \left(\xi_0(0),\zeta_0(0)\right) = (x_0,v_0).
    \end{cases}
\end{equation}
This describes the motion of a test particle without encountering a collision. The collisions are taken into account by introducing the subsequent branches of the process associated to the equations above. 
\begin{definition}[The Generalized Boltzmann Process]
    Let $(x_0,v_0)\in \mathbb{R}^2 \times S_1$ be the initial configuration for $t>0$. For $m\geq 1$ collisions, we define the generalized Boltzmann process
    \begin{equation}
\left(\xi_m(s),\zeta_m(s)\right) = \gamma^{-t+s}_m(x_0,v_0), \qquad s \in \mathbb{R}
    \end{equation}
    iteratively through:
\begin{enumerate}
        \item Choose $t_1$ such that $0 \leq t -  t_1 \leq T$. The test particle's trajectory $\left(\xi_m(s),\zeta_m(s)\right)$ for $s \in [t_1,t]$ is given by applying $\gamma_0$ until $t_1$ to $(x_0,v_0)$.
        \item Let $m \geq i \geq 1$. We consider $\zeta_m(t_i^+)$ and pick $n_i \in S_1$ such that
        \begin{equation}
            \zeta_m(t_i^+) \cdot n_i \geq 0.
        \end{equation}
        We flip the velocity $\zeta_m(t_i^+)$ by applying to it the the scattering map $S^{(1)}_{n_i}$. This defines the pre-collisional velocity $\zeta_m(t_i^-)$ or equivalently a rotation angle $\vartheta_i$. We go to the next step.
        \item Choose $t_{i+1}$ in such a way that $(t_{i+1}-t_i) (s-t) \geq 0$ and define
        \begin{equation}
            k_i := \left[(t_{i} - t_{i+1})/T\right]. 
        \end{equation}
        If $k_i = 0$, go directly to the next step; otherwise apply $\gamma_0$ until the first completion of the Larmor orbit (i.e. up to $t_i-T$). Then, rotate the velocity by the angle $\vartheta_i$. Repeat this procedure $k_i$-times until the time $t_i- k_i T$ is reached. Go to the next step.
        \item Apply $\gamma_0$ from $t_i - k_i T$ up to time $t_{i+1}$. Return to the second step.
    \end{enumerate}
    If $\zeta_m(t_i^+) \cdot n_i < 0$, we set the rotation angle $\vartheta_i$ to zero and do not flip the velocity. Together with $m=0$, $\left(\xi_m(s),\zeta_m(s)\right)$ forms the generalized Boltzmann process (defined backwards in time).
\end{definition}
We now introduce all possible events that might appear as a realisation of the magnetic Lorentz process:
\begin{definition}[(Self-)Recollisions] \label{def:(Self)Recollisions}
    Collisions for times different from $t_i$ with the obstacle centered at $c_i$ are called recollisions. A recollision with the $i$-th obstacle at $\Tilde{t}> t_i$ is called a self-recollision if the previously hit scatterer is labelled by $i$. Let $Q \in \mathbb{N}$ denote the internal obstacles, then the set describing the event of a recollision (disregarding self-recollisions) is given by
    \begin{equation}\label{eq:pathR}
        R := \left\{\min_{i = 1,...,Q} \min_{j = i+2,...,Q} \inf_{\tau_j \leq s \leq t_{j+1}} \abs{\xi_m(s)- c_i} \leq \varepsilon \right\}.
    \end{equation}
\end{definition}

\begin{definition}[Daisies]
    A self-recollision at $\Tilde{t}_i> t_i$ with the $i$-th obstacle happens when the test particle circles around the scatterer. The pieces of the cyclotron orbits will be called leaves and together with the obstacle we refer to the event as a daisy. We denote by $2 \beta^\varepsilon_i$ the angle, which is formed by the lines connecting $c_i$ with the two distinct points of impact $\xi^\varepsilon(t_i)$ and $\xi^\varepsilon(\Tilde{t}_i)$ at the obstacle's surface. Then, simple geometrical considerations lead to
\begin{equation} \label{eq: Self Recollisions}
    \cos \beta^\varepsilon_i = \frac{(\Delta_i^\varepsilon)^2 - R^2 + \varepsilon^2}{2 \Delta_i^\varepsilon \varepsilon},
\end{equation}
where $\Delta_i^\varepsilon \in (R-\varepsilon,R+\varepsilon)$.
\end{definition}
\begin{definition}[Periodic Daisies] \label{def:PeriodicDaisies}
    A periodic daisy is the result of $\beta^\varepsilon_i$ in~\eqref{eq: Self Recollisions} being a rational multiple of $2 \pi$. In that case, the test particle is trapped for ever at the obstacle $i$. We denote the event of a periodic daisy by 
    \begin{equation}\label{eq:pathPD}
        A := \bigcup_{m = 2}^K A_m^i,
\end{equation} 
where 
\begin{equation}\label{eq:pathPD1}
    A_m^i := \left\{\exists c_i: \frac{\beta_i^\varepsilon}{2 \pi} \in \mathbb{Q} \right\}.
\end{equation}
The set $ A_m^i$ represents the event that an obstacle centered at $c_i$ is surrounded by exactly one $m$-daisy. We would like to emphasize that the maximum number of leaves $K \geq 2$ is finite. This reflects that $\abs{B} < \infty$ (see also Remark~\ref{rem:B-infty}).
\end{definition}

\begin{definition}[Interferences] \label{def:Interferences} Given a collection of internal obstacles $Q \in \mathbb{N}$, then the set 
    \begin{equation} \label{eq:EventInterferences}
    I  = \left \{ \min_{i = 0,...,Q-1} \min_{j = i+2,...,Q} \inf_{\tau_i \leq s \leq t_{i+1}} \abs{\xi_m(s)- c_j} \leq \varepsilon \right\}
\end{equation}
represents the insertion of the $i$-th obstacle into the tube formed by the test particle before the first and unique impact time $t_i$. We will call the event described by $I$ interferences.
\end{definition}

\subsection{Periodic Daisies}
Let us consider a self-recollision of the test particle with the same obstacle $i$. We denote by $\{q_m\}$ the sequence of the centers of Lamor orbits.  Given an angle of incidence $\psi$, the test particle is reflected elastically and shifts the center of the cyclotron orbit to $q_1$ without changing the Lamor radius $R$. After a certain time, the test particle recollides with the obstacle $i$.
As already mentioned in Definition~\ref{def:PeriodicDaisies}, it can happen that the ratio between $\beta^\varepsilon_i$ and $2\pi$ is rational. In that case, the test particle forms a periodic $m$-daisy given by the centers of Lamor orbits $\{q_m\}$, where $2\leq m \leq K$ with $K\geq 2$ being a finite natural number.
The geometrical interpretation lets us conclude that a periodic $m$-daisy spans a regular polygon with $m$ corners labelled by $q_m$. An illustration of a periodic $3$-daisy can be found in Figure~\ref{fig:periodic3daisy}. 

\noindent We show that the periodic daisies vanish as $\varepsilon \to 0$ by proving

\begin{proposition} \label{prop:Daisies}
    Let $A$ be the event defined as in \eqref{eq:pathPD}-\eqref{eq:pathPD1}, namely $A$ is the event that a periodic daisy occurs. For any finite time $t$ and magnetic field, we have 
    \begin{equation}
        \mathbb{E}^\varepsilon [\mathbbm{ 1}_A] \leq C t \varepsilon \eta_\varepsilon,
    \end{equation}
    where $C>0$ independent of $t$.
\end{proposition}
\begin{proof}
We start by observing again that $2 \leq K < \infty$, otherwise $\beta_i^\varepsilon$ in~\eqref{eq: Self Recollisions} can be made arbitrarily small and the light particle would not circle around the obstacle. Let $A_0$ denote the event that the test particle is initially trapped in a periodic $m$-daisy.
Therefore, 
\begin{equation}
    \begin{aligned}
        \mathbb{E}^\varepsilon [\mathbbm{1}_{A_0}] &= \sum_{N \geq 0} \sum_{m = 2}^K e^{-\mu_\varepsilon} \int \frac{\mu_\varepsilon^N}{N!}  \mathbbm{1}_{A_m^i} \dd \bm{c}_N\\
        & \leq \sum_{N\geq 0 } \sum_{m = 2}^K e^{-\mu_\varepsilon} \int \frac{\mu_\varepsilon^N}{N!} N  \mathbbm{1}_{A_m^1} \dd {c}_1\\
        & \leq \sum_{N\geq 0 } e^{-\mu_\varepsilon} \frac{\mu_\varepsilon^N}{N!} N \sum_{m = 2}^K \pi \varepsilon^2 \\
        & \leq \pi \varepsilon^2 \mu_\varepsilon (K-1)\\
        & \leq \pi \varepsilon \eta_\varepsilon \mu (K-1).
    \end{aligned}
\end{equation}
In order to obtain the second inequality, we noted that $c_1$ must be located inside a disk of radius $\varepsilon$; otherwise the configuration would not lead to a periodic $m$-daisy. As time increases, it becomes more likely for the test particle to be trapped. We conclude that
\begin{equation}
    \mathbb{P}^\varepsilon(A)  \leq \sup_{s \in [0,t]} \mathbb{E}^\varepsilon [\mathbbm{1}_{A_0}] \leq C t \varepsilon \eta_\varepsilon.
\end{equation}
\end{proof}
\begin{remark}\label{rem:B-infty}
    Stating that the number of leaves $K$ must be finite is equivalent to assert that, as 
$\abs{B}\to\infty$, the system is expected to exhibit no dynamics.
\end{remark}
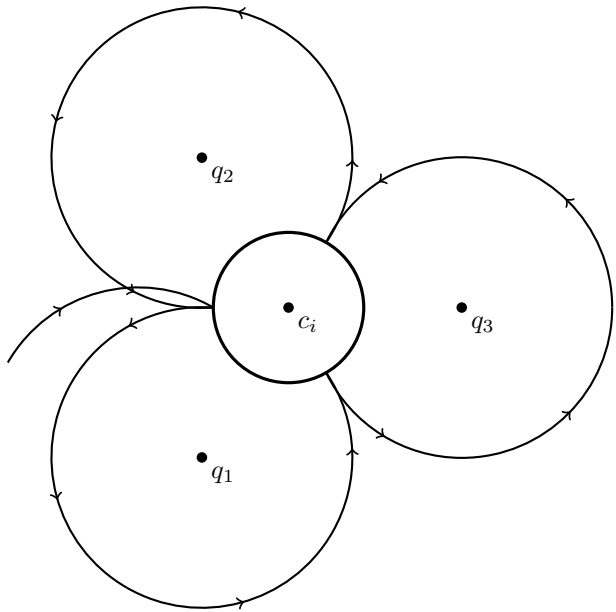
\begin{figure}
    \centering
\begin{tikzpicture}

\def \p{asin(sqrt(3)/4)};
\def \r{2};
\def \d{sqrt(\r^2- (sqrt(3)/2)^2)};

\draw[very thick] (0,0) circle (1);
 
\draw[thick,postaction={decorate, decoration={markings, mark=between positions 0.3 and 1.0 step 0.5 with {\arrow[]{>};}}}] ({\d+1/2},0) ++(2,0) arc[start angle=0, end angle={180-\p}, radius=\r];
\draw[thick,postaction={decorate, decoration={markings,mark=between positions 0.3 and 1.0 step 0.5 with {\arrow[]{<};}}}] ({\d+1/2},0) ++(2,0) arc[start angle=0, end angle={-(180-\p)}, radius=\r];
 
\draw[thick,rotate=120,postaction={decorate, decoration={markings,mark=between positions 0.3 and 1.0 step 0.5 with {\arrow[]{>};}}}] ({\d+1/2},0) ++(2,0) arc[start angle=0, end angle={180-\p}, radius=\r];
\draw[thick,rotate=120,postaction={decorate, decoration={markings,mark=between positions 0.3 and 1.0 step 0.5 with {\arrow[]{<};}}}] ({\d+1/2},0) ++(2,0) arc[start angle=0, end angle={-(180-\p)}, radius=\r];
 
\draw[thick,rotate=240,postaction={decorate, decoration={markings,mark=between positions 0.3 and 1.0 step 0.5 with {\arrow[]{>};}}}] ({\d+1/2},0) ++(2,0) arc[start angle=0, end angle={180-\p}, radius=\r];
\draw[thick,rotate=240,postaction={decorate, decoration={markings, mark=between positions 0.3 and 1.0 step 0.5 with {\arrow[]{<};}}}] ({\d+1/2},0) ++(2,0) arc[start angle=0, end angle={-(180-\p)}, radius=\r];

\draw[thick,postaction={decorate, decoration={markings, mark=between positions 0.7 and 1 step 0.9 with {\arrow[]{<};}}}] (-1,0)  arc[start angle=60, end angle=150, radius=\r];

        \coordinate[label=below right:{$q_{2}$}] (Q1) at ({cos(120)*(\d+1/2)},{sin(120)*(\d+1/2)});
        \coordinate[label=below right:{$q_{1}$}] (Q2) at ({cos(240)*(\d+1/2)},{sin(240)*(\d+1/2)});
        \coordinate[label=below right:{$q_{3}$}] (Q3) at ({\d+1/2},0);
        \coordinate[label=below right:{$c_{i}$}] (C) at (0,0);
        {\fill[] (C) circle (0.07);}
        {\fill[] (Q3) circle (0.07);}
        {\fill[] (Q2) circle (0.07);}
        {\fill[] (Q1) circle (0.07);}

\end{tikzpicture}
\caption{Illustration of a periodic $3$-daisy.}
\label{fig:periodic3daisy}
\end{figure}

\subsection{Recollisions}
Another event which spoils the Markovianity of the underlying process (see Remarks after Theorem~\ref{thm:MainTheorem}) are recollisions. Namely, the test particle collides again with a scatterer previously encountered.   We distinguish between two distinct kind of recollisions:  \textit{self-recollisions} and \textit{external recollisions} (see Definition~\ref{def:(Self)Recollisions}). 
While self-recollisions survive in the low-density limit, external recollisions are negligible as $\varepsilon\to 0$, as shown in the following estimate:
\begin{lemma} \label{lem:Recollisions}
    Let $R$ be the  event defined as in \eqref{eq:pathR}. For any finite time $t$ and magnetic field $B$, there exists $\varepsilon < \varepsilon_0$, such that the following estimate holds     \begin{equation}
        \mathbb{E}^\varepsilon [\mathbbm{ 1}_R] \leq C t^2 \varepsilon^{\frac{1}{2}} \eta_\varepsilon^3.
    \end{equation}
    The constant $C > 0$ is independent of $t$ and $B$.
\end{lemma}
\begin{proof}
Let $Q \in \mathbb{N}$ label the internal obstacles. 
We want to estimate the probability of $R$ during the dynamics.
We call $\alpha_{j,k}$ the relative scattering angle formed by the test particle's trajectory when leaving obstacle $j$ and hitting $k$ for the first time (see Definition~\ref{def:ScatteringData}). Therefore,
\begin{equation}
    \mathbb{E}^\varepsilon[\mathbbm{1}_R] = I^1_\varepsilon + I^2_\varepsilon,
\end{equation}
where
\begin{equation} \label{eq:Int1Recollision}
    \begin{aligned}
        I^1_\varepsilon =& e^{-2t \mu_\varepsilon \varepsilon} \sum_{Q\geq 0} \mu_\varepsilon^Q \sum_{i=1}^Q \sum_{j= i+2}^Q \int_0^t \dd t_1 \int_{t_1}^t \dd t_2 ... \int_{t_{Q-1}}^t \dd t_Q \\
        & \times\int_{-\varepsilon}^\varepsilon \dd b_1 \int_{-\varepsilon}^\varepsilon \dd b_2 ... \int_{-\varepsilon}^\varepsilon \dd b_Q \mathbbm{1}_{\left\{\inf_{\tau_j \leq s \leq t_{j+1}} \abs{\xi_m(s)- c_i} \leq \varepsilon\right\}}\\
        & \times \mathbbm{1}_{\left\{\sin \alpha_{j,k} \leq \frac{\varepsilon^\nu}{4},\, \forall k=i,...,j-1\right\}} 
    \end{aligned}
\end{equation}
and
\begin{equation} \label{eq:Int2Recollision}
    \begin{aligned}
        I^2_\varepsilon =& e^{-2t \mu_\varepsilon \varepsilon} \sum_{Q\geq 0} \mu_\varepsilon^Q \sum_{i=1}^Q \sum_{j= i+2}^Q \int_0^t \dd t_1 \int_{t_1}^t \dd t_2 ... \int_{t_{Q-1}}^t \dd t_Q \\
        & \times\int_{-\varepsilon}^\varepsilon \dd b_1 \int_{-\varepsilon}^\varepsilon \dd b_2 ... \int_{-\varepsilon}^\varepsilon \dd b_Q \mathbbm{1}_{\left\{\inf_{\tau_j \leq s \leq t_{j+1}} \abs{\xi_m(s)- c_i} \leq \varepsilon\right\}}\\
        & \times \mathbbm{1}_{\left\{\sin \alpha_{j,k} \geq \frac{\varepsilon^\nu}{4},\, \text{for some } k=i,...,j-1\right\}} .
    \end{aligned}
\end{equation}
The parameter $\nu \in (0,1)$ will be determined at the end of the calculation. We focus first on $I^1_\varepsilon$. 
In~\eqref{eq:Int1Recollision}, the second indicator implies that either all angles are close to zero or at least one is close $\pi$. 
In the case of $\alpha_{j,k} \leq \frac{\varepsilon^\nu}{4}$, a recollision is impossible, since the trajectory will never exit the cone with opening angle $\frac{\varepsilon^\nu}{2}$. 

\noindent We call $\theta_k$ the angle when summing up all self-recollisions with obstacle $k$ and $\varphi_k$ the angle after the $k$-th collision when the test particle follows a cyclotron orbit. Then, it can happen that $\varphi_k$ oriented with respect to $\theta_k$ is close to $\pi$. This justifies the following inequality
\begin{equation}
    \begin{aligned}
        I^1_\varepsilon  \leq &  e^{-2t \mu_\varepsilon \varepsilon} \sum_{Q\geq 0} \mu_\varepsilon^Q \sum_{i=1}^Q \sum_{j= i+2}^Q \sum_{k=i}^{j-1} \int_0^t \dd t_1 \int_{t_1}^t \dd t_2 ... \int_{t_{Q-1}}^t \dd t_Q \\
        & \times\int_{-\varepsilon}^\varepsilon \dd b_1 \int_{-\varepsilon}^\varepsilon \dd b_2 ... \int_{-\varepsilon}^\varepsilon \dd b_Q \mathbbm{1}_{\left\{\inf_{\tau_j \leq s \leq t_{j+1}} \abs{\xi_m(s)- c_i} \leq \varepsilon\right\}}\\
        & \times \mathbbm{1}_{\left\{\abs{\varphi_k -\pi} \leq \frac{\varepsilon^\nu}{4}\right\}}.\\
    \end{aligned}
\end{equation}
Indeed, $\abs{\varphi_k - \pi} \leq \frac{\varepsilon^\nu}{4}$ implies that the obstacle $c_{k+1}$ must be located inside a region $L$ formed by two displaced circles by $\frac{\varepsilon^\nu}{2}$ and having a radius $\varepsilon$. 
The mentioned region $L$ has a shape of a lune and therefore
\begin{equation}
    \begin{aligned}
        I^1_\varepsilon & \leq  e^{-2t \mu_\varepsilon \varepsilon} \sum_{Q\geq 0} \mu_\varepsilon^Q \frac{(2 \varepsilon t)^{Q-1}}{(Q-1)!} Q (Q-1) (Q-2) \int \mathbbm{1}_{\left\{c_1 \in L\right\}} \dd c_1\\
        &  \leq e^{-2t \mu_\varepsilon \varepsilon} \sum_{Q\geq 0} \mu_\varepsilon^Q \frac{(2 \varepsilon t)^{Q-1}}{(Q-3)!} Q  \abs{L}\\
        & \leq 4 \varepsilon^{\nu +1} 4t^2 \mu^3_\varepsilon \varepsilon^2 \\
        & \leq 16  t^2 \mu^3 \varepsilon^\nu \eta_\varepsilon^3,
    \end{aligned}
\end{equation}
where we estimated the Lebesgue measure of the lune by $\abs{L} \leq 4 \varepsilon^{\nu+1} $. \\
Let us proceed with $I^2_\varepsilon$.
The condition 
\begin{equation}
    \sin \alpha_{j,k} \geq \frac{\varepsilon^\nu}{4} 
\end{equation}
in~\eqref{eq:Int2Recollision} yields an estimate for the shortest path connecting $x_j^+:=\xi^\varepsilon(t_j^+)$ with $x_k^-:=\xi^\varepsilon(t_k^-)$. From the geometric interpretation, there exists a parallel line $l$ cutting the cylinder spanned by the internal obstacle $c_m$ in direction $c_{m-1}$. 
The parallel line $l$ can be oriented such that it cuts the cylinder with an angle $\frac{\alpha_{i,i+1}}{2}$ relative to the $x$-axis. 
In principle $\alpha_{i,i+1} \in (0,\pi)$ and therefore we use the property that 
\begin{equation}
    \forall x \in \left(0,\frac{2 \pi}{3}\right], \; \exists \varepsilon_0 > 0:\quad \sin \frac{x}{2} \geq \frac{\varepsilon_0^\nu}{4} 
\end{equation}
as well as 
\begin{equation}
    \forall x \in \left(\frac{2 \pi}{3}, \pi\right),\quad \sin (x) \leq \sin \frac{x}{2}
\end{equation}
in order to estimate
\begin{equation}
    \abs{x_j^+ - x_k^-} \leq \frac{2 \varepsilon}{\sin{(\alpha_{j,k} /2)}} \leq 8 \varepsilon^{1-\nu}.
\end{equation}
Hence,
\begin{equation}
    \abs{c_i - c_{i+1}} \leq 8 \varepsilon^{1-\nu} + 2 \varepsilon \leq 10 \varepsilon^{1-\nu}
\end{equation}
for $\varepsilon$ small enough. We can rephrase the statement by estimating the probability of finding a second scatterer in a box $J$ of measure (at most) $20 \varepsilon^{2-\nu}$
\begin{equation}
    \begin{aligned}
        I^2_\varepsilon &\leq  e^{-2t \mu_\varepsilon \varepsilon} \sum_{Q\geq 0} \mu_\varepsilon^Q \frac{(2 \varepsilon t)^{Q-1}}{(Q-1)!} Q(Q-1) (Q-2) \int  \mathbbm{1}_{\left\{c_1 \in J\right \}}\dd c_1\\
        & \leq 20 \varepsilon^{2-\nu} e^{-2t \mu_\varepsilon \varepsilon} \sum_{Q\geq 0} \mu_\varepsilon^Q \frac{(2 \varepsilon t)^{Q-1}}{(Q-3)!} Q\\
        &\leq 20  \varepsilon^{2-\nu} 4 t^2 \mu^3_\varepsilon \varepsilon^2\\
        & \leq 80  t^2 \mu^3 \varepsilon^{1-\nu}\eta_\varepsilon^3.
    \end{aligned}
\end{equation}
We conclude that $\nu = \frac{1}{2}$ optimises the rate of convergence.
\end{proof}

\subsection{Interferences}
 Interferences arise from placing an obstacle into an already existing tube before its unique impact time (see Definition~\ref{def:Interferences}).  Also this event is pathological and must be estimated. We show
\begin{lemma} \label{lem:Interferences}
    Let $I$ be the event defined as in \eqref{eq:EventInterferences}. For any finite time and magnetic field, there exists a pure constant $C>0$, such that the following estimate holds     \begin{equation}
        \mathbb{E}^\varepsilon[\mathbbm{1}_I] \leq C t^2 \varepsilon^{\frac{1}{2}} \eta_\varepsilon^3.
    \end{equation}
\end{lemma}
\begin{proof}
We keep the notation introduced in the proof of Lemma~\ref{lem:Recollisions} and refer to Definition~\ref{def:Interferences} for the set $I$. We define $B(x,r)$ to be the ball centered at $x \in \mathbb{R}^2$ with radius $r>0$.
If we show that
\begin{equation}
    \begin{aligned}
        \mathbb{E}^\varepsilon[\mathbbm{1}_I] \leq & e^{-2t \mu_\varepsilon \varepsilon} \sum_{Q\geq 1} \mu_\varepsilon^Q \sum_{i=0}^{Q-1}\sum_{j= i+1}^Q \int_0^t \dd t_1 \int_{t_1}^t \dd t_2 ... \int_{t_{Q-1}}^t \dd t_Q \\
        & \times\int_{-\varepsilon}^\varepsilon \dd b_1 \int_{-\varepsilon}^\varepsilon \dd b_2 ... \int_{-\varepsilon}^\varepsilon \dd b_Q \mathbbm{1}_{\left\{c_j \in \bigcup_{s\in (\tau_i, t_{i+1})} B(\xi_m(s),2\varepsilon)\right\}}
    \end{aligned}
\end{equation}
vanishes as $\varepsilon \to 0$, we can also deduce that the event in~\eqref{eq:EventInterferences} is negligible, since it is contained in the expression above.
Let us define
\begin{equation}
    d(x, \mathbb{Z}^*) := \inf_{k \in \mathbb{Z} \setminus \{0\}} \abs{x-k}
\end{equation}
in order to introduce the following splitting
\begin{equation}
    \mathbb{E}^\varepsilon[\mathbbm{1}_I] = D^1_\varepsilon + D^2_\varepsilon
\end{equation}
with
\begin{equation}
    \begin{aligned}
        D^1_\varepsilon =& e^{-2t \mu_\varepsilon \varepsilon} \sum_{Q\geq 1} \mu_\varepsilon^Q \sum_{i=0}^{Q-1}\sum_{j= i+1}^Q \int_0^t \dd t_1 \int_{t_1}^t \dd t_2 ... \int_{t_{Q-1}}^t \dd t_Q \\
        & \times\int_{-\varepsilon}^\varepsilon \dd b_1 \int_{-\varepsilon}^\varepsilon \dd b_2 ... \int_{-\varepsilon}^\varepsilon \dd b_Q \mathbbm{1}_{\left\{c_j \in \bigcup_{s\in (\tau_i, t_{i+1})} B(\xi_m(s),2\varepsilon)\right\}}\\
        &\times \mathbbm{1}_{\left\{d(\theta_{i+1} +\varphi_{i+1} + ... + \theta_{j-1} + \varphi_{j-1},\pi \mathbb{Z}^*)\geq \varepsilon^\delta\right\}} 
    \end{aligned}
\end{equation}
as well as
\begin{equation}
    \begin{aligned}
        D^2_\varepsilon =& e^{-2t \mu_\varepsilon \varepsilon} \sum_{Q\geq 1} \mu_\varepsilon^Q \sum_{i=0}^{Q-1}\sum_{j= i+1}^Q \int_0^t \dd t_1 \int_{t_1}^t \dd t_2 ... \int_{t_{Q-1}}^t \dd t_Q \\
        & \times\int_{-\varepsilon}^\varepsilon \dd b_1 \int_{-\varepsilon}^\varepsilon \dd b_2 ... \int_{-\varepsilon}^\varepsilon \dd b_Q \mathbbm{1}_{\left\{c_j \in \bigcup_{s\in (\tau_i, t_{i+1})} B(\xi_m(s),2\varepsilon)\right\}}\\
        &\times \mathbbm{1}_{\left\{d(\theta_{i+1} +\varphi_{i+1} + ... + \theta_{j-1} + \varphi_{j-1},\pi \mathbb{Z}^*) \leq \varepsilon^\delta\right\}}.
    \end{aligned}
\end{equation}
The parameter $\delta$ is determined again at the end of the estimate. We start with $D^1_\varepsilon$ and note that the \textit{impact} angle
\begin{equation}
    \alpha := d(\theta_{i+1} +\varphi_{i+1} + ... + \theta_{j-1} + \varphi_{j-1},\pi \mathbb{Z}^*)
\end{equation}
satisfies $\pi > \alpha \geq \varepsilon^\delta >0$. 
By assumption, the obstacle $c_j$ must lie in a tube with diameter $4 \varepsilon$ and therefore, the trajectory has to pass through a region of order $O(\varepsilon)$, so that we can bound the integral over $t_j$. 
This means that we consider the cone with opening angle $\pi-2 \alpha$, which is oriented with respect to the $x$-axis by the angle $\alpha$. 
Eventually, the constructed cone will cut the tube of the test particle in three pieces and the trajectory's segment can be bounded by the length of the sides $l_c$ forming the cone, which cuts the tube. We obtain
\begin{equation}
    2 l_c = 2 \frac{4 \varepsilon}{\sin \alpha} \leq 16 \varepsilon^{1-\delta},
\end{equation}
and therefore
\begin{equation}
    \begin{aligned}
        D^1_\varepsilon &\leq 16  \varepsilon^{1-\delta}e^{-2t \mu_\varepsilon \varepsilon} \sum_{Q\geq 1} (2 \varepsilon \mu_\varepsilon )^Q \frac{t^{Q-1}}{(Q-1)!}Q^2\\
        & \leq 16  t^2 \varepsilon^{1-\delta} (2 \varepsilon \mu_\varepsilon)^3\\
        & \leq 128  t^2 \mu^3 \varepsilon^{1-\delta} \eta_\varepsilon^3.
    \end{aligned}
\end{equation}
Before we study $D^2_\varepsilon$, we write down the relation between the deflection angle $\theta$ and the impact parameter $b$ in the case of hard disks.
We have that
\begin{equation}
    \sin \phi = \frac{b}{\varepsilon},
\end{equation}
where $\phi$ is the angle formed by $-\zeta$ when hitting the obstacle and the unit vector $n$.
Basic geometry yields 
\begin{equation}
    2 \phi + \theta = \pi,
\end{equation}
so that
\begin{equation}
    b = \varepsilon \sin \left(\frac{\pi}{2} - \frac{\theta}{2}\right) = \varepsilon \cos \frac{\theta}{2}.
\end{equation}
The (differential) scattering cross section $\Phi$ is the absolute value of the derivative of $b$ with respect to $\theta$. We see immediately that $\Phi$ is bounded by
\begin{equation}
    \Phi = \abs{\frac{\dd b}{\dd \theta} } = \frac{\varepsilon}{2} \abs{\sin \frac{\theta}{2}} \leq \frac{\varepsilon}{2}.
\end{equation}
By introducing the  set
\begin{equation}
    \mathbb{Z}^*_Q := \{-Q,-Q+1,..., -1,1,...,Q-1,Q\},
\end{equation}
we can rewrite $D^2_\varepsilon$ as follows 
\begin{equation}
    \begin{aligned}
        D^2_\varepsilon \leq &  e^{-2t \mu_\varepsilon \varepsilon} \sum_{Q\geq 1} \frac{(2t\mu_\varepsilon \varepsilon)^Q}{Q!} Q^2 \int_{-\pi}^\pi \dd \theta_{i+1}  ... \int_{-\pi}^\pi \dd \theta_{j-1} \frac{\Phi(\theta_{i+1})}{2 \varepsilon} ... \frac{\Phi(\theta_{j-1})}{2 \varepsilon}  \\
        &\times \mathbbm{1}_{\left\{d(\theta_{i+1} +\varphi_{i+1} + ... + \theta_{j-1} + \varphi_{j-1},\pi \mathbb{Z}^*_Q) \leq \varepsilon^\delta\right\}}.
    \end{aligned}
\end{equation}
By taking the supremum over the possible choices of $i$ as well as $j$ and using the conditions on the angles, we obtain
\begin{equation}
    \begin{aligned}
        D^2_\varepsilon & \leq  e^{-2t \mu_\varepsilon \varepsilon} \sum_{Q\geq 1} \frac{(2t\mu_\varepsilon \varepsilon)^Q}{Q!} Q^2  \varepsilon^\delta\\
        & \leq 2  (2 t \mu_\varepsilon \varepsilon)^2 \varepsilon^\delta\\
        & \leq 8  t^2 \mu^2 \varepsilon^\delta \eta_\varepsilon^2. 
    \end{aligned}
\end{equation}
Optimizing in $\delta$ we obtain $\delta =\frac{1}{2}$. 
\end{proof}

\section{Functional Properties of $\mathcal{L}^G$}\label{sec:invertibility}

The proof relies on the functional properties of the generalized Boltzmann operator introduced in~\eqref{eq:LinMagBoltzmann}. First, we present a method on how to invert $\mathcal{L}^G$ in terms of a Neumann series. In order to guarantee the convergence of the Neumann series, we have to impose a restriction on the strength of the magnetic field. 
 After introducing the inverse of $\mathcal{L}^G$, we obtain an expression that separates the Markovian part from the magnetic contribution. This formula is used to determine the diffusion coefficient.

\noindent We start by showing that $\mathcal{L}^G$ is invertible.

\begin{lemma} \label{lem: Inversion LG}
    For any $\abs{B} \in [0, \frac{8 \pi}{3})$ and $g \in L^\infty (S_1)$, such that $\int_{S_1} g(v) \dd v =0$, we have 
    \begin{equation}\label{eq:operator-norm-L-inverse}
        \norm{\left(\mathcal{L}^G\right)^{-1} g}_{L^\infty(S_1)} \leq \norm{g}_{L^\infty(S_1)}.
    \end{equation}
\end{lemma}
\begin{proof}
    We use the splitting of $\mathcal{L}^G$ as introduced in the previous section. We recall that the action of the linear Boltzmann operator on $f$ is given by
    \begin{equation} \label{eq:LinBoltzOp}
        \mathcal{L} f = \mu \int_{-1}^1 [f(v') - f(v)] \dd b,
    \end{equation}
    where we denote the impact parameter by $b \in [-1,1]$.
    The linear Boltzmann operator in~\eqref{eq:LinBoltzOp} suggests an additional splitting of $\mathcal{L}$ in the following way:
    \begin{equation}
        \mathcal{L} f = \mu \int_{-1}^1 [f(v') - f(v)] \dd b = \mu \int_{-1}^1 f(v') \dd b - 2 \mu f(v) = 2 \mu (K - \Id) f,
    \end{equation}
    where 
    \begin{equation}
        Kf = \frac{1}{2} \int_{-1}^1 f(v') \dd b
    \end{equation}
    by simple comparison with~\eqref{eq:LinBoltzOp}.
    Let now $ g \in L^\infty(S_1)$, such that $\int_{S_1} g(v) \dd v = 0$. Under this assumption, the expression $\mathcal{L}^G h = g$ implies
    \begin{equation}
        [2 \mu (K - \Id) + \mathcal{M}] h = g
    \end{equation}
    or equivalently
    \begin{equation} \label{eq:RecursionFormula}
        h = - \frac{1}{2\mu} g + \left[ K + \frac{1}{2 \mu}\mathcal{M}\right]h.
    \end{equation}
    Here we implicitly assume that $h \in L^1(\mathbb{R}^2 \times S_1)$ for~\eqref{eq:RecursionFormula} to make sense.
    By iterating the recursion formula in~\eqref{eq:RecursionFormula}, we obtain
    \begin{equation} \label{eq:NeumannSeries}
        h = -\frac{1}{2\mu} \sum_{n= 0}^\infty \left(K + \frac{1}{2 \mu} \mathcal{M}\right)^n g.
    \end{equation}
    Let us denote the operator norm by $\norm{\cdot}$. It is clear that the Neumann series in~\eqref{eq:NeumannSeries} converges if $\norm{K + \mathcal{M}/(2\mu)} < 1$. Under the assumption that $g \in L^{\infty}(S_{1})$ with $\int_{S_{1}}  g(v) \dd v= 0$, we know from~\cite{basile_derivation_2015} that 
    \begin{equation}
        \norm{K g}_{L^\infty(S_1)} = \norm{\frac{1}{2} \int_{-1}^1 g(v') \dd b}_{L^\infty(S_1)} \leq \frac{1}{2} (\pi-2) \norm{g}_{L^\infty(S_1)}
    \end{equation}
and from now on $\beta:=(\pi-2)/2$. Notice that $\beta< 1$.\\
We now bound $\norm{\mathcal{M}g}_{L^\infty(S_1)}$:
\begin{equation} \label{eq:Bound Magnetic Operator}
    \begin{aligned}
    &\norm{\mathcal{M}g}_{L^\infty(S_1)}\\
    &\leq \mu \sum_{k=1}^{[t/T]} e^{-2kT} \norm{\int_{S_1}  (v\cdot n )_+ \left \{\sigma_n - 1\right\} g(t-kT,S_n^{(k)}(x,v)) \dd n}_{L^\infty(S_1)}\\
    & \leq \mu \sum_{k=1}^{[t/T]} e^{-2kT} \left\{\norm{\int_{S_1}  (v\cdot n )_+\sigma_n g(t-kT,S_n^{(k)}(x,v))\dd n}_{L^\infty(S_1)} + 2 \norm{g}_{L^\infty(S_1)}\right\}\\
    & \leq \mu \sum_{k=1}^{[t/T]} e^{-2kT} \left\{  (\pi-2)\norm{g}_{L^\infty(S_1)} + 2 \norm{g}_{L^\infty(S_1)}\right\}\\
    & \leq \mu \frac{e^{-2T}}{1-e^{-2T}} \{2 \beta + 2\} \norm{g}_{L^\infty(S_1)},
    \end{aligned}
\end{equation}
where we identified the same structure in $\mathcal{M}$ as in $\mathcal{L}$ and interpreted $g(t-kT, S_n^{(n)}(x,v))$ as a new initial datum and $\sigma_n g(t-kT, S_n^{(n)}(x,v))$ as the usual flipping of velocity when encountering a collision. Furthermore, applying the formula for the geometric series is justified, since for any $T>0$ the expression $e^{-2T} < 1$ (see also Remark~\ref{rem:B-infty}).

\noindent We go back to the initial problem, where we impose $\norm{K + \mathcal{M}/(2\mu)} < 1$ in order to have a converging series in~\eqref{eq:NeumannSeries}. Looking at the operator norm, we have
\begin{equation} \label{eq:NormOperators}
    \norm{K + \frac{1}{2\mu}\mathcal{M}} \leq \norm{K} + \frac{1}{2\mu} \norm{\mathcal{M}} < 1.
\end{equation}
This implies a restriction on possible values for $T$. Indeed, if we consider~\eqref{eq:NormOperators} and use the estimates from above, it is equivalent to ensure that 
\begin{equation}
    \beta + \frac{e^{-2T}}{1-e^{-2T}} \{ \beta + 1\} < 1,
\end{equation}
which leads us finally to 
\begin{equation} \label{eq: Bound T}
    T = \frac{2 \pi}{\abs{B}} > \frac{1}{2} \log \left(\frac{2}{1-\beta}\right) > \frac{3}{4}.
\end{equation}
The statement follows from~\eqref{eq: Bound T}.
\end{proof}

We use the following Proposition to determine the explicit expression of the diffusion coefficient $D_B$ (see~\eqref{eq:D-Split}).  In particular, Proposition~\ref{prop:D-explicit} highlights the dependence of the diffusion coefficient on the magnetic field and on the non-Markovian components of the underlying stochastic process. We proceed as in the proof of Lemma~\ref{lem: Inversion LG}.
\begin{proposition}\label{prop:D-explicit}
    As long as $\abs{B} \in [0, \frac{8 \pi}{3})$, the following formula for the inverse of the generalized Boltzmann operator $\left(\mathcal{L}^G\right)^{-1} = (\mathcal{L} + \mathcal{M})^{-1}$ holds true  
    \begin{equation}
        \left(\mathcal{L}^G\right)^{-1} = \sum_{k=0}^\infty  \mathcal{L}^{-1} \left[\mathcal{M} (-\mathcal{L})^{-1}\right]^k.
    \end{equation}
\end{proposition}
\begin{proof}
    We proceed similarly as in Lemma~\ref{lem: Inversion LG}. Let us consider again $ h \in L^1(\mathbb{R}^2 \times S_1)$ and $g \in L^\infty (S_1)$, such that $\int_{S_1} g(v) \dd v=0 $. We start from $\mathcal{L}^G h = g$ and obtain
    \begin{equation}
        \mathcal{L}(\Id + \mathcal{L}^{-1} \mathcal{M})h =g.
    \end{equation}
    By multiplying with the inverse of $\mathcal{L}$ from left and collecting all terms including operators on the right, we have
    \begin{equation} \label{eq: Iteration h}
        h =  \mathcal{L}^{-1} g - \mathcal{L}^{-1}\mathcal{M} h.
    \end{equation}
    This recursion formula leads us to
    \begin{equation}
        h = \sum_{k=0}^\infty \mathcal{L}^{-1} \left[\mathcal{M} (-\mathcal{L})^{-1}\right]^k g
    \end{equation}
    after iterating~\eqref{eq: Iteration h}. For this expression to converge we impose
    \begin{equation}
        \norm{\mathcal{M} \mathcal{L}^{-1}} \leq \norm{\mathcal{M}} \norm{\mathcal{L}^{-1}} < 1,
    \end{equation}
    which is equivalent to
    \begin{equation}
        \norm{\mathcal{L}^{-1}} < \norm{\mathcal{M}}^{-1}.
    \end{equation}
    As we have already used $\norm{\mathcal{L}^{-1}} \leq \frac{1}{2 \mu (1-\beta)}$ in the proof of Lemma~\ref{lem: Inversion LG} and obtained~\eqref{eq:Bound Magnetic Operator}, we continue by 
    \begin{equation}
        \frac{1}{2 \mu (1 - \beta)} < \norm{\mathcal{M}}^{-1} = \left(\mu \frac{e^{-2T}}{1-e^{-2T}} \{2 \beta + 2\}\right)^{-1}.
    \end{equation}
    Indeed, the expression above gives the same bound on $T$ as in~\eqref{eq: Bound T}. This concludes the proof.
\end{proof}
\section{Proof of the Main Theorem} \label{sec:Proof}

\noindent In order to prove Theorem~\ref{thm:MainTheorem}, we study first the hydrodynamic limit of the generalized Boltzmann equation in~\eqref{eq:LinMagBoltzmann}. As it will be obvious later on, it is more convenient to introduce the following Cauchy problem to avoid confusions
\begin{equation} 
    \begin{cases} \label{eq:RescaledLinearLandauCauchy}
        (\partial_t + \eta_\varepsilon v \cdot \nabla_x - \eta_\varepsilon (v \times B) \cdot \nabla_v)g^\varepsilon(t,x,v) = \eta_\varepsilon^2 \mathcal{L}^G g^\varepsilon(t,x,v)\\
        g^\varepsilon(t=0,x,v) = f_0(x,v).
    \end{cases}
\end{equation}
We perform a Hilbert expansion of $g^\varepsilon$ and obtain the heat equation in the leading order with $D_B$ given by~\eqref{eq:D-Split}. Theorem~\ref{thm:MainTheorem} is concluded by commenting on the validity of the error estimates for longer time scales. 

\subsection{The Hilbert Expansion}

Although we are interested in the hydrodynamic regime of~\eqref{eq:RescaledLinearLandauCauchy}, we prove first the following preliminary result
\begin{proposition}\label{prop:ConvergenceToAverage}
    Let $\langle g^\varepsilon \rangle := \frac{1}{2 \pi} \int_{S_1}  g^\varepsilon(t,x,v)\dd v$ with $g^\varepsilon$ satisfying~\eqref{eq:RescaledLinearLandauCauchy}. Under the assumptions of Theorem~\ref{thm:MainTheorem}, for all $t \in (0,\bar{t}\,]$, 
    \begin{equation}
        g^\varepsilon - \langle g^\varepsilon \rangle \xrightarrow[]{\varepsilon \to 0} 0 \quad \text{in} \quad L^2(\mathbb{R}^2 \times S_1).
    \end{equation}
    If we define $t_{\eta_\varepsilon} := \frac{1}{\eta^\omega_\varepsilon}$ with $\omega >2$, we have that
    \begin{equation}
        g^\varepsilon(t_{\eta_\varepsilon}) - \langle f_0\rangle \xrightarrow[]{\varepsilon \to 0} 0 \quad \text{in} \quad L^2(\mathbb{R}^2 \times S_1),
    \end{equation}
    where $\langle f_0\rangle := \frac{1}{2 \pi}\int_{S_1}  f_0(x,v)\dd v$.
\end{proposition}
\begin{proof}
    We introduce $R_{\eta_\varepsilon} := g^\varepsilon - \langle g^\varepsilon \rangle$. 
    By applying the PDE in~\eqref{eq:RescaledLinearLandauCauchy} to $R_{\eta_\varepsilon}$, we obtain
    \begin{equation}
        \begin{aligned}
            &(\partial_t + \eta_\varepsilon v \cdot \nabla_x - \eta_\varepsilon  (v \times B) \cdot \nabla_v) R_{\eta_\varepsilon} \\
            &= (\partial_t + \eta_\varepsilon v \cdot \nabla_x - \eta_\varepsilon  (v \times B) \cdot \nabla_v)(g^\varepsilon - \langle g^\varepsilon \rangle)\\
            &= \eta^2_\varepsilon \mathcal{L}^G g^\varepsilon - (\partial_t + \eta_\varepsilon v \cdot \nabla_x - \eta_\varepsilon  (v \times B) \cdot \nabla_v) \langle g^\varepsilon \rangle\\
            &= \eta^2_\varepsilon \mathcal{L}^G ( g^\varepsilon - \langle g^\varepsilon \rangle) - (\partial_t + \eta_\varepsilon v \cdot \nabla_x) \langle g^\varepsilon \rangle,
        \end{aligned}
    \end{equation}
    where the last line follows after noting that $\langle g^\varepsilon \rangle$ is not a function of $v$. Therefore, the expression from above is equivalent to 
    \begin{equation}
        (\partial_t + \eta_\varepsilon v \cdot \nabla_x - \eta_\varepsilon  (v \times B) \cdot \nabla_v) R_{\eta_\varepsilon} = \eta^2_\varepsilon \mathcal{L}^G R_{\eta_\varepsilon} + \varphi
    \end{equation}
    with
    \begin{equation}
        \begin{aligned}
            \varphi&:=-(\eta_\varepsilon v \cdot \nabla_x \langle g^\varepsilon \rangle + \partial_t \langle g^\varepsilon \rangle) = -\left(\eta_\varepsilon v \cdot \nabla_x \langle g^\varepsilon \rangle + \int_{S_1}\partial_t g^\varepsilon \dd v\right)\\
            &= -\eta_\varepsilon v \cdot \nabla_x \langle g^\varepsilon \rangle +\frac{1}{2\pi}\int_{S_1}\eta_\varepsilon (v \cdot \nabla_x -  (v \times B) \cdot \nabla_v)g^\varepsilon\dd v -  \frac{\eta_\varepsilon^2}{2 \pi}\int_{S_1} \mathcal{L}^G g^\varepsilon \dd v \\
            &=\eta_\varepsilon \left(\frac{1}{2 \pi}\int_{S_1}v \cdot \nabla_x g^\varepsilon\dd v - v \cdot \nabla_x \langle g^\varepsilon \rangle \right).
        \end{aligned}
    \end{equation}
    The last equality follows directly from the divergence theorem applied to $\langle \mathcal{L}^G g^\varepsilon \rangle$. For $\hat{n}$ being the the outward pointing unit normal on $S_1$ and $\dd \sigma$ the corresponding surface measure, we have
    \begin{equation} \label{eq:DivTheoremLandauOperator}
        \int_{S_1} \mathcal{L}^G g^\varepsilon \dd v = \int_{\partial S_1} \nabla_\abs{v} g^\varepsilon \cdot \hat{n} \dd \sigma = 0,
    \end{equation}
    since $\partial S_1 = \varnothing$. 
    By a similar reasoning, the external force contributing by $-\eta_\varepsilon  (v \times B) \cdot \nabla_v g^\varepsilon$ integrated with respect to $v$ over $S_1$ vanishes too. 
    Therefore, $\varphi$ can be uniformly bounded in time:
    \begin{equation}
        \sup_{t \leq \bar{t}} \norm{\varphi}_{L^2}  \leq \sup_{t \leq \bar{t}}  C \eta_\varepsilon \norm{\nabla_x g^\varepsilon }_{L^2}  \leq \eta_\varepsilon C \bar{t}.
    \end{equation}
    Let us denote the inner product in $L^2$ by $(\cdot,\cdot)$. Then, we obtain the following equation for $R_{\eta_\varepsilon}$
    \begin{equation}\label{eq:ODE Reta}
        \begin{aligned}
            \frac{1}{2} \frac{\dd }{\dd t} \norm{R_{\eta_\varepsilon}(t)}^2_{L^2} &= \eta_\varepsilon^2 (R_{\eta_\varepsilon},\mathcal{L}^G R_{\eta_\varepsilon}) + (R_{\eta_\varepsilon},\varphi) \\
            &\leq -\eta_\varepsilon^2(R_{\eta_\varepsilon},-\mathcal{L}^G R_{\eta_\varepsilon}) + \norm{R_{\eta_\varepsilon}}_{L^2} \norm{\varphi}_{L^2}\\
            &\leq - \eta_\varepsilon^2 \lambda \norm{R_{\eta_\varepsilon}}^2_{L^2} + \norm{R_{\eta_\varepsilon}}_{L^2} \norm{\varphi}_{L^2}.
        \end{aligned}
    \end{equation}
    In the last line from above we use the positivity of $-\mathcal{L}^G$ and introduce the smallest eigenvalue $\lambda>0$.
    By solving~\eqref{eq:ODE Reta}, we find
    \begin{equation}
        \begin{aligned}
            \norm{R_{\eta_\varepsilon}(t)}_{L^2} &\leq e^{-\lambda\eta_\varepsilon^2 t} \norm{R_{\eta_\varepsilon}(0)}_{L^2} + \int_0^t  e^{-\lambda \eta_\varepsilon^2 (t-s)} \norm{\varphi (s)}_{L^2} \dd s\\
            &\leq e^{-\lambda\eta_\varepsilon^2 t} \norm{R_{\eta_\varepsilon}(0)}_{L^2} + \frac{C}{\lambda \eta_\varepsilon} \left(1-e^{-\lambda \eta_\varepsilon^2 t}\right)
        \end{aligned}
    \end{equation}
    and thus, $\norm{R_{\eta_\varepsilon}(t)}_{L^2} \xrightarrow[]{\varepsilon \to 0}0$ for all $t \in (0,\bar{t}]$, which proves the first part of the claim. 
    We proceed with
    \begin{equation}
        \begin{aligned}
            \frac{1}{2} \frac{\dd}{\dd t} \norm{g^\varepsilon (t) - f_0}^2_{L^2} &= (g^\varepsilon - f_0, (\partial_t + \eta_\varepsilon v \cdot \nabla_x - \eta_\varepsilon (v \times B) \cdot \nabla_v)(g^\varepsilon-f_0))\\
            &= (g^\varepsilon - f_0,(\partial_t + \eta_\varepsilon v \cdot \nabla_x - \eta_\varepsilon (v \times B) \cdot \nabla_v)g^\varepsilon)\\
            &- (g^\varepsilon - f_0,( \eta_\varepsilon v \cdot \nabla_x - \eta_\varepsilon (v \times B) \cdot \nabla_v)f_0),
        \end{aligned}
    \end{equation}
    where we used that $f_0$ does not depend on time. We end up with
    \begin{equation} \label{eq:ClosenessTof0 ODE}
        \begin{aligned}
            &\frac{1}{2} \frac{\dd}{\dd t} \norm{g^\varepsilon (t) - f_0}^2_{L^2} \\
            &= (g^\varepsilon - f_0 , \eta_\varepsilon^2 \mathcal{L}^G g^\varepsilon) - \eta_\varepsilon (g^\varepsilon - f_0 ,  v \cdot \nabla_x f_0) + \eta_\varepsilon(g^\varepsilon - f_0 , (v \times B) \cdot \nabla_v f_0)\\
            &\leq \eta_\varepsilon^2 (g^\varepsilon - f_0 , \mathcal{L}^G f_0)- \eta_\varepsilon (g^\varepsilon - f_0 ,  v \cdot \nabla_x f_0) + \eta_\varepsilon (g^\varepsilon - f_0 , (v \times B) \cdot \nabla_v f_0)\\
            &\leq \norm{g^\varepsilon - f_0}_{L^2} \left(\eta_\varepsilon  \norm{\nabla_x f_0}_{L^2}+ \eta_\varepsilon \abs{B} \norm{\nabla_v f_0}_{L^2} + \eta_\varepsilon^2 \norm{\mathcal{L}^G f_0}_{L^2}\right).
        \end{aligned}
    \end{equation}
    We solve~\eqref{eq:ClosenessTof0 ODE} and obtain
    \begin{equation}
        \norm{g^\varepsilon(t) - f_0}_{L^2} \leq t \left(\eta_\varepsilon  \norm{\nabla_x f_0}_{L^2}+ \eta_\varepsilon \abs{B} \norm{\nabla_v f_0}_{L^2} + \eta_\varepsilon^2 \norm{\mathcal{L}^G f_0}_{L^2}\right),
    \end{equation}
    where we have set $g^\varepsilon(0) = f_0$. 
    The proof is concluded by 
    \begin{equation}
        \begin{aligned}
            \norm{g^\varepsilon(t_{\eta_\varepsilon}) - \langle f_0 \rangle}_{L^2} &\leq \norm{g^\varepsilon(t_{\eta_\varepsilon}) - \langle g^\varepsilon (t_{\eta_\varepsilon})\rangle}_{L^2} + \norm{ \langle g^\varepsilon (t_{\eta_\varepsilon})\rangle - \langle f_0 \rangle}_{L^2}\\
            &\leq \sup_{0<t \leq \bar{t}} \norm{g^\varepsilon(t) - \langle g^\varepsilon (t)\rangle}_{L^2}+ \norm{ \langle g^\varepsilon (t_{\eta_\varepsilon}) - f_0 \rangle}_{L^2}\\
            &\leq \sup_{0<t \leq \bar{t}} \norm{g^\varepsilon(t) - \langle g^\varepsilon (t)\rangle}_{L^2}+ \langle \norm{  g^\varepsilon (t_{\eta_\varepsilon}) - f_0 }_{L^2} \rangle,
        \end{aligned}
    \end{equation}
    which implies
    \begin{equation}
        \lim_{\varepsilon\to 0}\norm{g^\varepsilon (t_{\eta_\varepsilon}) - \langle f_0 \rangle}_{L^2(\mathbb{R}^2 \times S_1)} = 0 .
    \end{equation} 
\end{proof}

\noindent In the next proposition we introduce the Hilbert expansion of $g^\varepsilon$. The following result ensures that $g^\varepsilon$ and $\rho$ satisfying the heat equation in~\eqref{eq:HeatEquation} have the same asymptotics.
\begin{proposition} \label{prop: Hilbert Expansion}
    Let $g^\varepsilon$ solve~\eqref{eq:RescaledLinearLandauCauchy}. Then, for any $t \in [0,\bar{t}\,]$ and $\abs{B} \in \left[0, \frac{8 \pi}{3}\right)$ 
    \begin{equation}
        \lim_{\varepsilon\to 0}\norm{g^\varepsilon(t,\cdot,\cdot) - \rho(t,\cdot)}_{L^2(\mathbb{R}^2 \times S_1)} = 0,
    \end{equation}
    where $\rho$ satisfies the heat equation in~\eqref{eq:HeatEquation}. The diffusion coefficient $D_B$ is given by the Green-Kubo relation in~\eqref{eq:GreenKubo} and can be expressed as in~\eqref{eq:D-Split}.
\end{proposition}
\begin{proof}
    We start with the truncated Hilbert expansion of $g^\varepsilon$:
    \begin{equation} \label{eq: Truncated Hilbert}
        g^\varepsilon(t,x,v) = g^{(0)} (t,x,v) + \frac{1}{\eta_\varepsilon} g^{(1)}(t,x,v) + \frac{1}{\eta_\varepsilon^2} g^{(2)}(t,x,v) + \frac{1}{\eta_\varepsilon}R_{\eta_\varepsilon},
    \end{equation}
    where $R_{\eta_\varepsilon}$ denotes the reminder of the series.
    We insert the expression~\eqref{eq: Truncated Hilbert} into~\eqref{eq:RescaledLinearLandauCauchy} and obtain
    \begin{equation} \label{eq:HilbertInPDE}
        \begin{aligned}
            &(\partial_t + \eta_\varepsilon v \cdot \nabla_x - \eta_\varepsilon (v \times B) \cdot \nabla_v)g^\varepsilon\\
            &= (\partial_t + \eta_\varepsilon v \cdot \nabla_x - \eta_\varepsilon (v \times B) \cdot \nabla_v)\left(g^{(0)} + \frac{1}{\eta_\varepsilon} g^{(1)}+ \frac{1}{\eta_\varepsilon^2} g^{(2)} + \frac{1}{\eta_\varepsilon}R_{\eta_\varepsilon}\right)\\
            &= \eta_\varepsilon^2 \mathcal{L}^G g^\varepsilon = \eta_\varepsilon^2 \mathcal{L}^G \left(g^{(0)} + \frac{1}{\eta_\varepsilon} g^{(1)}+ \frac{1}{\eta_\varepsilon^2} g^{(2)} + \frac{1}{\eta_\varepsilon}R_{\eta_\varepsilon}\right).
        \end{aligned}
    \end{equation}
    By organising~\eqref{eq:HilbertInPDE} with respect to the powers of $\eta_\varepsilon$ and defining $A_{\eta_\varepsilon}(t):= \partial_t g^{(1)} + \frac{1}{\eta_\varepsilon}\partial_t g^{(2)}+v\cdot\nabla_x g^{(2)} - (v \times B) \cdot \nabla_v  g^{(2)}$, we have the the following set of equations
    \begin{enumerate}[label={(\roman*)}]
        \item \label{item:one} $\mathcal{L}^G g^{(0)} = 0$,
        \item \label{item:two} $v \cdot \nabla_x g^{(0)} -(v \times B) \cdot \nabla_v g^{(0)} = \mathcal{L}^G g^{(1)}$,
        \item $\partial_t g^{(0)}+ v \cdot \nabla_x g^{(1)}-(v \times B) \cdot \nabla_v g^{(1)} = \mathcal{L}^G g^{(2)}$
        \item $(\partial_t + \eta_\varepsilon v \cdot \nabla_x - \eta_\varepsilon (v \times B) \cdot \nabla_v) R_{\eta_\varepsilon} = \eta_\varepsilon^2 \mathcal{L}^G R_{\eta_\varepsilon} - A_{\eta_\varepsilon}(t)$.
    \end{enumerate}
    The equation~\ref{item:one} yields that $g^{(0)}$ must lie in the null space of the generalized linear Boltzmann operator. Notice that functions satisfying this condition are constants with respect to the velocity variable.    
    As a consequence of~\ref{item:one}, item~\ref{item:two} has only a solution if the left hand side belongs to 
    \begin{equation}
    	\left(\mathrm{ker} \left(\mathcal{L}^G\right)\right)^{\perp} := \left\{h \in L^{2}(S_1 ): \int_{S_1} h(v) \dd v = 0 \right\}.
    \end{equation}
        By integrating~\ref{item:two} with respect to $v$, we obtain
    \begin{equation}
        \int_{S_1}  \left( v \cdot \nabla_x -(v \times B) \cdot \nabla_v \right) g^{(0)}  \dd v =  \int_{S_1} v \cdot \nabla_x  g^{(0)}  \dd v =0,
    \end{equation}
    since $g^{(0)} $ does not depend on $v$ and therefore $ v \cdot \nabla_x  g^{(0)} $ is an odd function with respect to the velocity variable. In Section~\ref{sec:invertibility} we have shown that  $\mathcal{L}^G$ is invertible, hence
    \begin{equation} \label{eq:Identity g1}
        g^{(1)} = \left(\mathcal{L}^G\right)^{-1} v \cdot \nabla_x g^{(0)} . 
    \end{equation}
Notice that, since $g^{(0)}\in L^{2}(S_1)$ and $\left(\mathcal{L}^G\right)^{-1}$ is bounded in operator norm (see~\eqref{eq:operator-norm-L-inverse}), $g^{(1)}\in L^2(S^1)$, i.e.
    \begin{equation*}
    \norm{g^{(1)}}_{L^2}\leq C\norm{\left(\mathcal{L}^G\right)^{-1}}\norm{\nabla_x g^{(0)}}_{L^2}.
    \end{equation*}
    For the third item of the list, we find that the right hand side vanishes:
    \begin{equation} \label{eq: Integral v Lg2}
        \int_{S_1}\mathcal{L}^G g^{(2)} \dd v = \int_{\partial S_1} \nabla_\abs{v} g^{(2)} \cdot \hat{n} \dd \sigma  = 0
    \end{equation}
    by a similar argument as in~\eqref{eq:DivTheoremLandauOperator}. If we integrate the left hand side with respect to $v$ and use the result from above, we get
    \begin{equation} \label{eq:Heat g0}
        \partial_t g^{(0)} + \frac{1}{2 \pi}\int_{S_1}v \cdot \nabla_x g^{(1)} \dd v = 0
    \end{equation}
    after applying the divergence theorem to $(v \times B) \cdot \nabla_v g^{(1)}$.
    We continue by inserting the expression for $g^{(1)}$ in~\eqref{eq:Identity g1} into~\eqref{eq:Heat g0}:
    \begin{equation}
        \begin{aligned}
            &\partial_t g^{(0)} + \frac{1}{2 \pi}\int_{S_1} v \cdot \nabla_x \left \{ \left(\mathcal{L}^G\right)^{-1} v \cdot \nabla_x g^{(0)} \right \} \dd v\\
            &= \partial_t g^{(0)} + \frac{1}{2 \pi}\int_{S_1} v \cdot  \left(\mathcal{L}^G\right)^{-1} v \Delta_x g^{(0)}  \dd v =0.
        \end{aligned}
    \end{equation}
    We define the diffusion coefficient as
    \begin{equation} \label{eq:DiffusionIndex-Form}
        D_B^{ij}:=  \frac{1}{2 \pi}\int_{S_1}  v_i \left(-\mathcal{L}^G\right)^{-1} v_j \dd v
    \end{equation}
    and note by symmetry that $D_B^{ij} = D_B \delta_{ij}$. 
    By the negativity of $\left(\mathcal{L}^G\right)^{-1}$ follows $D_B>0$ and we derive the Green-Kubo relation from~\eqref{eq:DiffusionIndex-Form}:
    \begin{equation}
        \begin{aligned}
            D_B= \frac{1}{2\pi} \int_{S_1}  v \cdot \left(-\mathcal{L}^G\right)^{-1} v \dd v &= \frac{1}{2 \pi}\int_{S_1} \int_0^\infty v \cdot e^{\mathcal{L}^G t} v \dd t \dd v\\
            & =  \int_0^\infty \mathbb{E} [v \cdot V_t(v)] \dd t.
        \end{aligned}
    \end{equation}
    With expression~\eqref{eq:DiffusionIndex-Form}, equation~\eqref{eq:Heat g0} simplifies to
    \begin{equation}
        \partial_t g^{(0)} - D_B \Delta_x  g^{(0)}  = 0.
    \end{equation}
    Let $g^{(0)}$ satisfy the initial condition $g^\varepsilon(t=0,x,v) = g^{(0)}(0,x,v)= g^{(0)}(0,x)$. Clearly, $g^{(0)} \in L^2$, since it satisfies the heat equation and by~\eqref{eq:Identity g1}, the regularity of $g^{(0)}$ also implies that $g^{(1)}\in L^2$.

\noindent     In~\eqref{eq: Integral v Lg2} we proved that $\int_{S_1} \mathcal{L}^G g^{(2)} \dd v =0$. Hence,
    \begin{equation} \label{eq:Derivative v g(2)}
        \begin{aligned}
            g^{(2)} = &\left(\mathcal{L}^G\right)^{-1} \left(\partial_t g^{(0)}+ v \cdot \nabla_x g^{(1)}-(v \times B) \cdot \nabla_v g^{(1)} \right)\\
            = &\left(\mathcal{L}^G\right)^{-1} \bigg(D_B \Delta_x  g^{(0)} + v \cdot \left(\mathcal{L}^G\right)^{-1} v \Delta_x g^{(0)}\\
            &-(v \times B) \cdot \nabla_v \left \{\left(\mathcal{L}^G\right)^{-1} v \cdot \nabla_x g^{(0)} \right \}\bigg).\\
        \end{aligned}
    \end{equation}
   By the same argument as above, we conclude that the $L^2$-norm of $g^{(2)}$ stays bounded as well.

\noindent    The proof is complete if we show that the $L^2$-norm of $R_{\eta_\varepsilon}$ is finite for any time $t\in [0,\bar{t}\,]$. We proceed with
    \begin{equation}
        \begin{aligned}
            \frac{1}{2} \frac{\dd}{\dd t} \norm{R_{\eta_\varepsilon}(t)}^2_{L^2} &= - \eta_\varepsilon^2 (R_{\eta_\varepsilon}, - \mathcal{L}^G R_{\eta_\varepsilon}) - (R_{\eta_\varepsilon}, A_{\eta_\varepsilon}(t))\\
            &\leq - \lambda \eta_\varepsilon^2 \norm{R_{\eta_\varepsilon}}^2_{L^2} + \norm{R_{\eta_\varepsilon}}_{L^2} \norm{A_{\eta_\varepsilon} (t)}_{L^2},
        \end{aligned}
    \end{equation}
    where we used the same argument as in~\eqref{eq:ODE Reta}. The ODE from above implies that 
    \begin{equation}
        \frac{\dd}{\dd t} \norm{R_{\eta_\varepsilon}(t)}_{L^2} \leq  \norm{A_{\eta_\varepsilon} (t)}_{L^2}.
    \end{equation}
    In order to find a suitable bound for $A_{\eta_\varepsilon}$ in $L^2$, we consider
    \begin{equation}
        \begin{aligned}
            \partial_t g^{(1)} &= \partial_t \left(\left(\mathcal{L}^G\right)^{-1}  v \cdot \nabla_x g^{(0)}\right)\\
            &= \left(\mathcal{L}^G\right)^{-1} \left( \Dot{v} \cdot \nabla_x g^{(0)} + v \cdot \nabla_x \partial_t g^{(0)}\right)\\
            &=\left(\mathcal{L}^G\right)^{-1} \left(\Dot{v} \cdot \nabla_x g^{(0)} + v \cdot \nabla_x \left \{D_B \Delta_x g^{(0)} \right \}\right).
        \end{aligned}
    \end{equation}
    By the last line, we see that $\partial_t g^{(1)}$ is finite in $L^2$. A similar result can be also obtained for $\partial_t g^{(2)}$, $v \cdot \nabla_x g^{(2)}$ as well as $-(v \times B) \cdot \nabla_v g^{(2)}$ in terms of derivatives of $g^{(0)}$ with respect to $x$ and $v$. With this we showed that $\norm{A_{\eta_\varepsilon}}_{L^2}$ is uniformly bounded for all $t \in [0,\bar{t}\,]$. More precisely,
    \begin{equation}
        \norm{R_{\eta_\varepsilon}(t)}_{L^2} \leq \int_0^{\bar{t}}  \norm{A_{\eta_\varepsilon}(s)}_{L^2} \dd s \leq C \bar{t}.
    \end{equation}
   Indeed, the Hilbert expansion in~\eqref{eq: Truncated Hilbert} converges to $g^{(0)} \equiv \rho$ as $\varepsilon \to 0$, which finalizes the proof.
    
\end{proof}
\begin{remark}
    It is sufficient to consider the truncated Hilbert expansion in~\eqref{eq: Truncated Hilbert}, since the general ansatz 
    \begin{equation}
        g^\varepsilon = \sum_{k=0}^\infty \eta^{-k}_\varepsilon g^{(k)}
    \end{equation}
    gives the recursion formula below
    \begin{equation}
        \partial_t g^{(n)} + v \cdot \nabla_x g^{(n+1)} - (v \times B) \cdot \nabla_v  g^{(n+1)} = \mathcal{L}^G g^{(n+2)}, \qquad n \geq 1.
    \end{equation}
    The method in the proof can be iterated for any $n$.
\end{remark}
\subsection{Conclusion of the Proof for Theorem~\ref{thm:MainTheorem}}
We conclude the proof by arguing on the convergence of $h_\varepsilon$ and $\rho$. We want to show
\begin{equation}
    \lim_{\varepsilon \to 0}\norm{h_\varepsilon(t) - \rho(t)}_{L^2} = 0.
\end{equation}
By the triangle inequality and the identity for $h_\varepsilon$ we obtain that
\begin{equation} \label{eq:TriangleConclusionProof}
    \begin{aligned}
    \norm{h_\varepsilon(t) - \rho(t)}_{L^2}  &= \norm{f^\varepsilon(\eta_\varepsilon t) - \rho(t)}_{L^2} \\
    &\leq \norm{f^\varepsilon(\eta_\varepsilon t) - g^\varepsilon(t)}_{L^2} +    \norm{g^\varepsilon(t) - \rho(t)}_{L^2}.
    \end{aligned}
\end{equation}
Indeed, $g^\varepsilon(t) \equiv \bar{f}^{G,\varepsilon}(\eta_\varepsilon t)$ is the solution to~\eqref{eq:TimeRescaledBoltzmann}. The last term in~\eqref{eq:TriangleConclusionProof} goes to zero as $\varepsilon \to 0$ according to Proposition~\ref{prop: Hilbert Expansion} and therefore we focus on controlling $\norm{f^\varepsilon(\eta_\varepsilon t) - g^\varepsilon(t)}_{L^2}$. By Proposition~\ref{prop:ErrorEstimates} and Remark~\ref{rem:kinetic-to-longer}, $f^\varepsilon(\eta_\varepsilon t)$ and $g^\varepsilon(t)$ converges in $L^1$ for any $t\in [0, \bar{t}]$ if 
\begin{equation}
    \varepsilon^{\frac{1}{2}}\eta_\varepsilon^5 \to 0
\end{equation}
as $\varepsilon \to 0$. Furthermore, Theorem~\ref{thm:MainTheorem} states that the initial datum $f_0$ has compact support and as a consequence $f^\varepsilon(\eta_\varepsilon t)$ as well as $g^\varepsilon(t)$ are also compactly supported for any $t \in [0,\bar{t}]$. This reasoning follows directly from~\eqref{eq:fep}. Hence, the convergence in $L^1$ implies further
\begin{equation}
    \norm{f^\varepsilon(\eta_\varepsilon t) - g^\varepsilon(t)}_{L^2} \xrightarrow[]{\varepsilon \to 0} 0,
\end{equation}
which concludes the proof of Theorem~\ref{thm:MainTheorem}.
\bigskip

{\bf Acknowledgments.} 
A.N.~and C.S.~are grateful to H.~Spohn for suggesting the reference~\cite{spohn_lorentz_1978}, which lead to the completion of this work.
D.N. and C.S. acknowledge the support of the Swiss National Science Foundation through the NCCR SwissMAP and the SNSF Eccellenza project PCEFP2\_181153, and by the Swiss State Secretariat for Research and Innovation through the ERC Starting Grant project P.530.1016 (AEQUA).
The research of A.N. has been supported by the project PRIN 2022 (Research Projects of National Relevance) - Project code 202277WX43.

\printbibliography

\end{document}